%% file: main.tex
\documentclass[letterpaper,twocolumn,10pt]{article}
\usepackage{usenix}

\usepackage{tikz}
\usepackage{amsmath}

\usepackage{filecontents}
\usepackage{marvosym}
\input{packages}

\input{notions}
\usepackage{threeparttable}

\usepackage[available,functional,reproduced]{usenixbadges}

\begin{document}

\date{}

\title{\Large \bf Membership Inference Attacks on Tokenizers of Large Language Models}


\author{
{\rm Meng Tong\textsuperscript{1\,*}}\ \ \ 
{\rm Yuntao Du\textsuperscript{2\,*}}\ \ \
{\rm Kejiang Chen\textsuperscript{1\,\href{mailto:chenkj@ustc.edu.cn}{\textcolor{black}{\Letter}}}}\ \ \
{\rm Weiming Zhang\textsuperscript{1}}\ \ \
{\rm Ninghui Li\textsuperscript{2}}\ \ \
\\
\\
\textsuperscript{1}\textit{University of Science and Technology of China} \ \ \ 
\textsuperscript{2}\textit{Purdue University} \ \ \ \\
\textsuperscript{*}\textit{Equal Contribution} \ \ \ 
\textsuperscript{\href{mailto:chenkj@ustc.edu.cn}{\textcolor{black}{\Letter}}}\textit{Corresponding Author} \ \ \
}

\maketitle
\begin{abstract}
Membership inference attacks (MIAs) are widely used to assess the privacy risks associated with machine learning models. However, when these attacks are applied to pre-trained large language models (LLMs), they encounter significant challenges, including mislabeled samples, distribution shifts, and discrepancies in model size between experimental and real-world settings. 
To address these limitations, we introduce tokenizers as a new attack vector for membership inference. Specifically, a tokenizer converts raw text into tokens for LLMs. Unlike full models, tokenizers can be efficiently trained from scratch, thereby avoiding the aforementioned challenges. In addition, the tokenizer’s training data is typically representative of the data used to pre-train LLMs.
Despite these advantages, the potential of tokenizers as an attack vector remains unexplored.
To this end, we present the first study on membership leakage through tokenizers and explore five attack methods to infer dataset membership. Extensive experiments on millions of Internet samples reveal the vulnerabilities in the tokenizers of state-of-the-art LLMs. To mitigate this emerging risk, we further propose an adaptive defense. Our findings highlight tokenizers as an overlooked yet critical privacy threat, underscoring the urgent need for privacy-preserving mechanisms specifically designed for them.

\end{abstract}
\input{tex/1_Introduction}
\input{tex/1.2_Background}
\input{tex/2_threat_models}

\input{tex/3.1_MIA_via_Merge_Similarity}

\input{tex/3.2_MIA_via_Vocabulary_Overlap}

\input{tex/3.3_MIA_via_Frequency_Estimation}

\input{tex/4_Experiments}

\input{tex/5._Related_Work}

\input{tex/6._Discussion}
\input{tex/7._Conclusion}


\appendix
\input{tex/Ethic_Consideration}

\input{tex/Open_Science}


\bibliographystyle{plain}
\bibliography{jobname}

\input{tex/appendix}

\end{document}

%% file: packages.tex
\usepackage{epsfig,endnotes,xspace,url,diagbox}

\usepackage{amsmath, amsthm, amsfonts}
\usepackage{dsfont}
\usepackage{enumitem}
\usepackage{bbm}
\usepackage{varioref}
\usepackage{graphicx}
\usepackage{subfigure}
\usepackage{pifont}

\usepackage{caption}
\usepackage{float}
\usepackage{cleveref}

\usepackage{epstopdf}
\usepackage{xcolor}
\usepackage{colortbl} 

\usepackage{multirow}
\usepackage{comment}
\usepackage{color}
\usepackage[utf8]{inputenc}
\usepackage{mathtools}

\usepackage{wrapfig}
\usepackage{tabularx}
\usepackage{xspace}
\usepackage{algorithm}
\usepackage{algorithmic}
\usepackage{cite}
\usepackage{booktabs}
\usepackage[skins]{tcolorbox}
\usepackage{url}
\theoremstyle{plain}

\setlist[itemize]{leftmargin=*,noitemsep, topsep=0pt}

\usepackage{bm}
\usepackage{graphicx}
\usepackage{subfigure}
\usepackage{makecell}
\usepackage{wasysym}
\usepackage{float} 
\usepackage{placeins}
\newtheorem{example}{Example}[section]
\newtheorem{definition}[example]{Definition}

\newtheorem{theorem}[example]{Theorem}

\definecolor{darkred}{RGB}{139,0,0}

\newcommand{\msignal}{\ensuremath{\mathsf{Merge~Similarity}}\xspace}
\newcommand{\nsignal}{\ensuremath{\mathsf{Naive~Bayes}}\xspace}
\newcommand{\csignal}{\ensuremath{\mathsf{Compression~Rate}}\xspace}
\newcommand{\vsignal}{\ensuremath{\mathsf{Vocabulary~Overlap}}\xspace}
\newcommand{\fsignal}{\ensuremath{\mathsf{Frequency~Estimation}}\xspace}

%% file: notions.tex
\newcommand{\ie}{{i.e., }}
\newcommand{\eg}{{e.g., }}

\newcommand{\et}{{et al.}}
\newcommand{\mypara}[1]{\smallskip\noindent\textbf{#1.} \xspace}
\newcommand{\mysubpara}[1]{\noindent\textit{#1}. \xspace}

\newcommand{\algcomment}[1]{\hfill {\color{blue} $\triangleright$ \emph{\small{#1}}}}

\usepackage{todonotes}



%% file: tex/1_Introduction.tex
\section{Introduction}\label{sec: intro}
Scaling up the pre-training data for large language models (LLMs) has been shown to improve performance~\cite{muennighoff2023scaling, brown2020language, henighan2020scaling, kaplan2020scaling}. Nevertheless, the rapid expansion of pre-training data has also raised concerns about whether these commercial models are trained on sensitive or copyrighted information~\cite{united1976code, meeus2024sok}. For instance, on June 4, 2025, Reddit filed a lawsuit against Anthropic, alleging the unlawful use of data from its 100 million daily users to train LLMs~\cite{nytimes2025reddit}. Furthermore, an increasing body of research~\cite{25LLMV,carlini2021extracting,25Towards} has documented instances in which LLMs memorize and leak private information. 

To assess potential data misuse, extensive research has explored the membership inference attacks (MIAs) in LLMs~\cite{shokri2017membership,zarifzadeh2024low,pang2023white,tong2025membership}.
In particular, an MIA aims to determine whether a specific data sample or dataset was used to train the target model (i.e., \textit{member}) or not (i.e., \textit{non-member}). To achieve this, existing MIAs primarily rely on the model's output as the attack vector~\cite{duarte2024cop,25Towards,zhang2024min}.
Although this vector is widely adopted, these attacks face significant challenges in reliably demonstrating their effectiveness for LLMs~\cite{hayes2025strong, meeus2024sok}, as shown in Figure~\ref{fig: evaluation flaws}. The primary obstacle is that faithful evaluation~\cite{hoffmann2022training} requires an evaluator to pre-train an LLM from scratch~\cite{ye2022enhanced}, which incurs significant computational costs.
As a result, existing MIAs are typically evaluated using LLMs that have already been pre-trained by others. Nevertheless, this may lead to MIA evaluation exhibiting \textit{distribution shifts}~\cite{maini2024llm, duarte2024cop} or containing \textit{mislabeled samples}~\cite{meeus2024sok, shi2023detecting}. Furthermore, many of the evaluated models (e.g., Pythia-12B~\cite{biderman2023pythia}) are much smaller than practical deployed LLMs (e.g., DeepSeek-R1-671B~\cite{guo2025deepseek}), limiting the ability to assess current MIAs in real-world conditions. Given these challenges, a natural question arises: \textit{Can we exploit an attack vector for MIAs against LLMs that avoids these limitations?}

\begin{figure}[t]
\centering
\includegraphics[width=1\columnwidth]{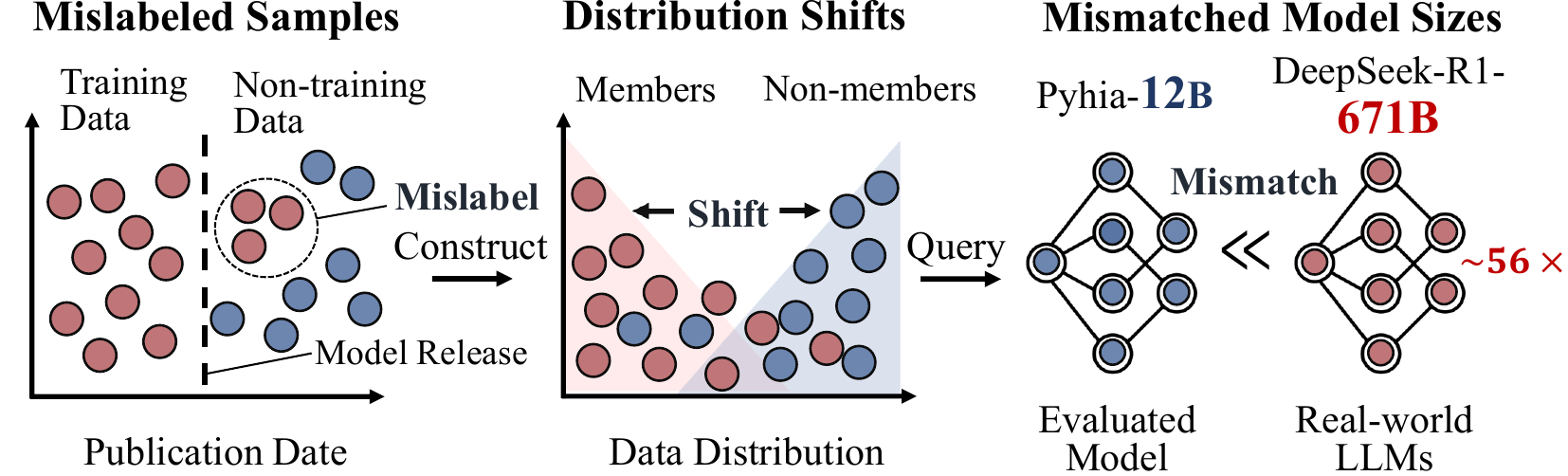}
\caption{Evaluation challenges in MIAs against LLMs.}
\label{fig: evaluation flaws}
\end{figure}


\mypara{New Attack Vector}
Motivated by this question, we explore a new attack vector that targets other components of LLMs.
Typically, an LLM comprises a tokenizer, a transformer network, and an output layer~\cite{touvron2023llama}. Among these components, the tokenizer has been open-sourced in commercial LLMs such as OpenAI-o3~\cite{openai2025tiktoken} and Gemini-1.5~\cite{karter2024gemini} to support transparent billing. Building on this observation, \textit{we propose the previously overlooked tokenizer as a new attack vector for membership inference.}
Specifically, a tokenizer~\cite{zouhar2023formal} is in charge of converting text into tokens for LLMs. Its training data is typically representative of the overall pre-training corpus of the LLMs~\cite{touvron2023llama,team2023gemini,brown2020language,parmar2024nemotron,dai2024deepseekmoe}. Its training process simply involves merging the most frequent strings into a vocabulary using the byte-pair encoding (BPE) algorithm~\cite{galle2019investigating}. This straightforward process enables training a tokenizer from scratch, which aligns with the inference game~\cite{ye2022enhanced} and avoids mislabeled samples or distribution shifts. Furthermore, the simplicity of BPE also makes it feasible to train a tokenizer that matches those used in state-of-the-art LLMs (see Figure~\ref{fig: cmp tokenizers}).

Despite these advantages, the feasibility of using tokenizers as an attack vector has not yet been explored. In this paper, we present the first study to exploit toknizers for MIAs and propose five attack methods for inferring dataset membership:
\begin{itemize}
\item \mypara{MIA via \msignal} This attack trains shadow tokenizers \cite{shokri2017membership} and compares their token merge orders to that of the target tokenizer. If the merge order of the target tokenizer closely matches that of shadow tokenizers trained on a particular dataset, the dataset is classified as a \textit{member}. However, the effectiveness of this attack is limited. Only a few distinctive tokens display a membership signal in merge order, making it difficult for membership inference.

\item \mypara{MIA via \vsignal} Leveraging these distinctive tokens, we propose a more effective attack, MIA via \vsignal. This attack also involves training shadow tokenizers. But instead of comparing merge orders, it classifies a dataset as a \textit{member} if the distinctive tokens in the target tokenizer’s vocabulary significantly overlap with those from shadow tokenizers trained on that dataset.

\item \mypara{MIA via \fsignal} While our results show that MIA via \vsignal achieves strong performance, it requires substantial time to train multiple shadow tokenizers. For efficient implementation, we introduce MIA via \fsignal, which trains only a single shadow tokenizer. This attack evaluates whether training the target tokenizer on a dataset is necessary for certain tokens to appear in its vocabulary. If this condition is met, this attack classifies the dataset as a \textit{member}.
\item  Additionally, as part of our evaluation in Section~\ref{setup}, we further explore two attack methods for potential alternatives: MIA via \nsignal and MIA via \csignal.

\end{itemize}

We conduct extensive evaluations using millions of Internet data~\cite{raffel2020exploring}. To match real-world practice, we align the trained tokenizers in evaluations with those used in commercial LLMs~\cite{touvron2023llama,biderman2023pythia,guo2025deepseek}. The experimental results indicate that MIAs via \vsignal and \fsignal achieve strong performance across various settings. 
For example, MIA via \vsignal achieves an AUC score of 0.771 against a tokenizer with two hundred thousand tokens, whereas MIA via \fsignal achieves an AUC score of 0.740. More importantly, our experiments show that scaling laws increase tokenizers' vulnerability to MIAs. This finding suggests that MIAs could become more effective on scaled-up tokenizers in the future.

\mypara{Our Contributions} Our main contributions are as follows:
\begin{itemize}
    \item We introduce the tokenizer as a new attack vector for membership inference and conduct the first study demonstrating its feasibility in LLMs.
    \item We explore five attack methods for set-level membership inference against tokenizers, revealing the vulnerabilities in these foundational components of state-of-the-art LLMs.
    \item We conduct extensive evaluations using real-world Internet datasets. The results show that our shadow-based attacks demonstrate strong performance against tokenizers.
    \item We further analyze tokenizers from commercial LLMs. The results show that the tokenizers, such as OpenAI-o200k~\cite{openai2025tiktoken} and DeepSeek-R1~\cite{guo2025deepseek}, also contain distinctive tokens for implementing membership inference.
\end{itemize}

\mypara{Main Findings} We have the following key findings:
\begin{itemize}
    \item According to prior work~\textup{\cite{ huangover,mayilvahanan2025llms}}, scaling up the intelligence of LLMs involves expanding the tokenizer's vocabulary~\textup{\cite{tao2024scaling}} and thus improving its compression efficiency~\textup{\cite{liu-etal-2025-superbpe}}. However, our experimental results show that it also increases the tokenizer's vulnerability to effective MIAs.
    \item The membership status of the target dataset with more data samples is typically more accurately inferred by MIAs. 
    \item While removing infrequent tokens from the tokenizer’s vocabulary can partially reduce the effectiveness of MIAs, this approach also lowers the tokenizer’s compression efficiency. Moreover, even with this mitigation, MIAs can remain effective for inferring large datasets.
\end{itemize}

\mypara{Organization} The remainder of this paper is organized as follows. Section~\ref{bg} presents preliminaries on tokenizer training and membership inference. Section~\ref{threat_model} introduces the threat models for membership inference attacks. Section~\ref{mia_method} presents three of our attack methods. Section~\ref{mia_evaluation} introduces two additional methods and reports the experimental results. Section~\ref{related_work} reviews related work. Section~\ref{discussion} discusses the limitations of LLM  dataset inference. Section~\ref{conclusion} concludes the paper.

%% file: tex/1.2_Background.tex
\section{Preliminaries}\label{bg}
\subsection{Tokenizer Training}
A tokenizer~\cite{zouhar2023formal} is a fundamental component in LLMs, converting raw text into a format that the model can process. Formally, a tokenizer is defined as a function $f_{\mathcal{V}}: S \rightarrow \mathcal{V}^*$ that maps an input string $s \in S$ (\eg a sentence or document) into a sequence of tokens from a vocabulary $\mathcal{V}$.
In practice, this function is learned from a collection of text datasets $\mathcal{D}$. Specifically, its training objective is to segment and encode the data in a way that maximizes compression efficiency~\cite{zouhar2023formal}. This process begins by initializing the vocabulary $\mathcal{V}$ with basic symbols, such as individual characters.
During training, the tokenizer $f_{\mathcal{V}}$ iteratively merges the most frequent pairs of symbols in the data via the byte-pair encoding (BPE) algorithm~\cite{shibata1999byte}. This iterative process results in a token merge order: each token $t_i \in \mathcal{V}$ is assigned an index $i$ corresponding to the iteration, where it was merged into the vocabulary $\mathcal{V}$.

In commercial LLM applications, the tokenizer also serves as a basis for token billing. As the tokenizer directly determines how users are charged based on the number of tokens in a message, its operation is critical for ensuring transparent billing~\cite{huang2025efficient}. To promote such transparency in token counting, the organizations behind major LLMs~\cite{openai2025tiktoken,karter2024gemini,guo2025deepseek} have open-sourced their tokenizers, making their vocabularies and token merge orders publicly available.

\subsection{Membership Inference}
The concept of membership inference attacks was first introduced by Shokri \et~\cite{shokri2017membership}, who demonstrated that an adversary can determine whether a specific data record was included in a model's training set. Specifically, they propose to train multiple shadow models that imitate the behavior of the target model. By comparing the output distributions of shadow models trained with and without a specific data record, the adversary can infer whether the data record was part of the target model's training data~\cite{nasr2019comprehensive}.

Building upon this foundational research, subsequent studies~\cite{carrington2022deep,carlini2021extracting,ndss26cpmia} have investigated the effectiveness of MIAs on a variety of machine learning models, including ResNet-18~\cite{he2016deep} and BERT~\cite{devlin2019bert}. Nonetheless, as models increase in size and are trained on larger datasets over fewer epochs, the overfitting signal for individual samples decreases, resulting in reduced MIA performance on LLMs~\cite{duanmembership}. To address this limitation, recent MIAs~\cite{25LLMV} instead focus on dataset membership, which aggregates signals from individual samples to enhance the detection of membership. 

%% file: tex/2_threat_models.tex
\section{Threat Models}\label{threat_model}

\begin{table}[t]
    \centering
    \caption{Tokenizer Information of Commercial LLMs.}
    \label{Tokenizer Justification}
    \resizebox{\columnwidth}{!}{
    \begin{threeparttable}  

    \vspace{-0.1cm}
    
    \begin{tabular}{cccc}
    \toprule 
    \textbf{LLMs} & \textbf{Affiliation} & \textbf{Tokenizer's Training Data} & \textbf{Download Tokenizer} \\
    \midrule 
    LLaMA~\cite{touvron2023llama} & Meta & LLaMA's Pretraining Data & \href{https://github.com/meta-llama/llama}{GitHub} \\
    Gemini~\cite{team2023gemini} & Google & Gemini's Pretraining Data & \href{https://github.com/googleapis/go-genai/blob/2c046453716ca9c10d445da5c0923b5b170773f2/tokenizer/tokenizer.go}{GitHub} \\
    GPT-3\tnote{1}~\hspace{0.07cm}\cite{brown2020language} & OpenAI & GPT-3's Pretraining Data & \href{https://github.com/openai/tiktoken/blob/main/tiktoken_ext/openai_public.py}{GitHub} \\
    Nemotron\tnote{2}~\hspace{0.07cm}\cite{parmar2024nemotron} & NVIDIA & Nemotron's Pretraining Data & \href{https://huggingface.co/nvidia/Nemotron-4-340B-Instruct}{Hugging Face} \\
    DeepSeek~\cite{dai2024deepseekmoe} & DeepSeek & DeepSeek's Pretraining Data & \href{https://huggingface.co/deepseek-ai/deepseek-moe-16b-base}{Hugging Face} \\
    \bottomrule
    \end{tabular}

    \begin{tablenotes}
      \item[1] GPT-3 utilizes the GPT-2's tokenization, i.e., tokenizer, which is trained on the WebText~\cite{radford2019language} (one of GPT-3's pretraining corpora).
      \item[2] The Series of Nemotron~\cite{adler2024nemotron,parmar2024nemotron} share both the pretraining data and the tokenizer.
    \end{tablenotes}
    \end{threeparttable} 
        }
        \vspace{-0.1cm}
\end{table}

\mypara{Attack Scenario} We consider an attack scenario in which an adversary implements MIAs to infer the pretraining datasets of an LLM. As demonstrated in Table~\ref{Tokenizer Justification}, the tokenizer for such an LLM is trained on the model's pretraining data. The adversary exploits the tokenizer as an attack vector for MIAs. Specifically, it downloads the tokenizer from official links, which were originally intended to support transparent billing for the service. Then, the adversary captures the distinctive tokens in the tokenizer to calculate the membership signal and infer the final result. Unlike previous MIAs against LLMs which directly reflect the model's memorization of pretraining data~\cite{25Towards}, MIAs on tokenizers lead to privacy implications, as illustrated by the case below:
\begin{itemize}
    \item[] \mysubpara{Case for Privacy Infringement} A distinctive token ‘davidjl’ from a Reddit user’s data appears in OpenAI’s tokenizers~\cite{hn36242914, reddittop}. Indeed, OpenAI claims that its tokenizers were trained on models' pretraining corpus, which includes Reddit data~\cite{radford2019language, brown2020language}. By leveraging distinctive tokens, MIAs on such tokenizers can infer whether users' data was used to train the tokenizers, thereby revealing the models' pretraining corpus and identifying privacy infringements.
\end{itemize}

\mypara{Adversary's Objective}
Given a set of pretraining datasets $\mathcal{D}_\textup{mem}$ sampled from an underlying distribution \(\mathbb{D}\) (denoted as $\mathcal{D}_\textup{mem}\leftarrow \mathbb{D}$), we write $\mathcal{V}_\text{target} \leftarrow \mathcal{T}(\mathcal{D}_\textup{mem})$ to represent a tokenizer's vocabulary $\mathcal{V}_\text{target}$ is trained by running the BPE algorithm $\mathcal{T}$~\cite{galle2019investigating} on $\mathcal{D}_\textup{mem}$. This training process results in the target tokenizer $f_{\mathcal{V}_\text{target}}$.
Given a target dataset $D \in \mathbb{D}$, the adversary's objective is to determine whether $D$ was part of the pretraining data used to construct the vocabulary $\mathcal{V}_\text{target}$. To achieve this, the adversary employs a membership inference attack $\mathcal{A}$,  which can be formally defined as:
\begin{equation}
\mathcal{A} \colon D, f_{\mathcal{V}_\text{target}} \rightarrow \{0,1\}, 
\end{equation}
where $1$ indicates \(D \in \mathcal{D}_\textup{mem}\), and $0$ indicates \(D \notin \mathcal{D}_\textup{mem}\).

\mypara{Adversary's Capabilities} In alignment with the attack scenario, where commercial LLMs open-source their tokenizers~\cite{openai2025tiktoken,karter2024gemini}, we assume that the adversary has access to the target tokenizer $f_{\mathcal{V}_\text{target}}$ and its token vocabulary $\mathcal{V}_\text{target} = \{t_1, t_2, \dots, t_{|\mathcal{V}_\text{target}|}\}$, where each token $t_i \in \mathcal{V}_\text{target}$ was merged at iteration $i$ during the training process. Furthermore, we assume that the adversary is able to sample auxiliary datasets $\mathcal{D}_\textup{aux}$ from the same distribution as the training data used by the target tokenizer, \ie $\mathcal{D}_{\text{aux}}\gets\mathbb{D}$. Leveraging datasets $\mathcal{D}_\textup{aux}$, the adversary can use the BPE algorithm $\mathcal{T}$ to train shadow tokenizers. This assumption is consistent with previous work~\cite{carlini2022membership, li2021membership, salem2018ml}. It is also realistic in practice, as the training data is primarily sourced from web content~\cite{guo2025deepseek}.

%% file: tex/3.1_MIA_via_Merge_Similarity.tex
\section{Attack Methodology}\label{mia_method}
In this section, we present our MIAs against pre-trained LLMs. 
For each method, we start by introducing our design intuition. Then we describe the attack methodology.

\subsection{\texorpdfstring{Baseline: MIA via \msignal}{MIA via Merge Similarity}}\label{initial_exploration}

Shadow-based  MIAs~\cite{hayes2025strong,hui2021practical,yeom2018privacy,li2021membership,salem2018ml} involve training auxiliary models to calibrate predictions. Inspired by these attacks, we formalize MIA via \msignal on tokenizers. 

\mypara{Design Intuition} Prior work~\cite{carlini2022membership} has revealed that the overfitting behavior of machine learning models can vary depending on whether a particular data point was present in the training data. Based on this insight, we hypothesize that tokenizers may also differ depending on whether a dataset was included in the training data. Specifically, we assume that the token merge order can serve as an indicator of the overfitting phenomenon. Thus, merge orders in vocabularies $\mathbb{V}_{\text{in}} = \{ {\mathcal{V}_\text{in}} \leftarrow \mathcal{T}(\mathcal{D}_\textup{aux}\bigcup\, \{D\}) \;|\; \mathcal{D}_\textup{aux} \leftarrow \mathbb{D} \}$ and $\mathbb{V}_{\text{out}}  = \{ {\mathcal{V}_\text{out}} \leftarrow \mathcal{T}(\mathcal{D}_\textup{aux} \setminus \{D\}) \;|\; \mathcal{D}_\textup{aux} \leftarrow \mathbb{D} \}$ can differ depending on whether the target dataset $D$ was included in the training data. Building on this hypothesis, an adversary can exploit this difference by comparing the similarity $\rho$ of token merge orders for pairs $(\mathcal{V}_\text{in},\mathcal{V}_\text{target})$ and $(\mathcal{V}_\text{out},\mathcal{V}_\text{target})$. If the average value for $\rho(\mathcal{V}_\text{in},\mathcal{V}_\text{target})$ is higher than that of $ \rho(\mathcal{V}_\text{out},\mathcal{V}_\text{target})$, it is more likely that $D\in\mathcal{D}_\textup{mem}$.

\mypara{Attack Method}
This attack consists of four steps.
\begin{enumerate}[label={(\roman*)}, leftmargin=1.8em, labelsep=0.3em]
    \item The adversary randomly samples a collection of datasets $\mathcal{D}_\textup{aux}\leftarrow \mathbb{D}$ for $N$ times, and trains $N$ shadow tokenizers considering inclusion or exclusion of the target dataset $D$. Thus, the adversary obtains sets $\mathbb{V}_{\text{in}}$ and $\mathbb{V}_{\text{out}}$.
    \item The adversary computes the similarity of token merge orders for each $\rho(\mathcal{V}_\text{in},\mathcal{V}_\text{target})$ and $\rho(\mathcal{V}_\text{out},\mathcal{V}_\text{target})$ via Spearman's rank correlation coefficient~\cite{sedgwick2014spearman}.
    \item The membership signal for target dataset $D$ is defined as:
    \begin{equation} 
      \hspace{-2.2em}  \frac{1}{2} + \frac{ \sum_{\mathcal{V}_\text{in}\in\mathbb{V}_{\text{in}}}{\rho(\mathcal{V}_\text{in},\mathcal{V}_\text{target})}}{4|\mathbb{V}_{\text{in}}|}- \frac{ \sum_{\mathcal{V}_\text{out}\in\mathbb{V}_{\text{out}}}\rho(\mathcal{V}_\text{out},\mathcal{V}_\text{target})}{4|\mathbb{V}_{\text{out}}|},
    \end{equation}
    where it ranges from 0 to 1. 
    \item If the membership signal is larger than a decision-making threshold $\tau$, output 1 (\textit{member}). Otherwise, output 0.
\end{enumerate}


\begin{figure}[t]
    \centering
    \subfigure[Vocabulary Size: $80\text{,\,}000$]{\includegraphics[width=0.49\columnwidth]{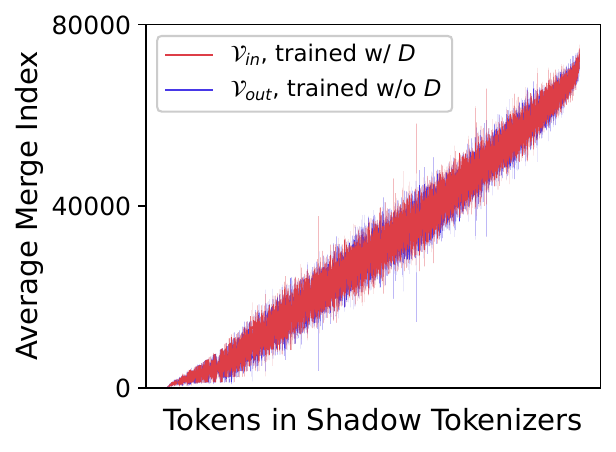}}
    \hspace{0.1em}\subfigure[Vocabulary Size: $110\text{,\,}000$]{\includegraphics[width=0.49\columnwidth]{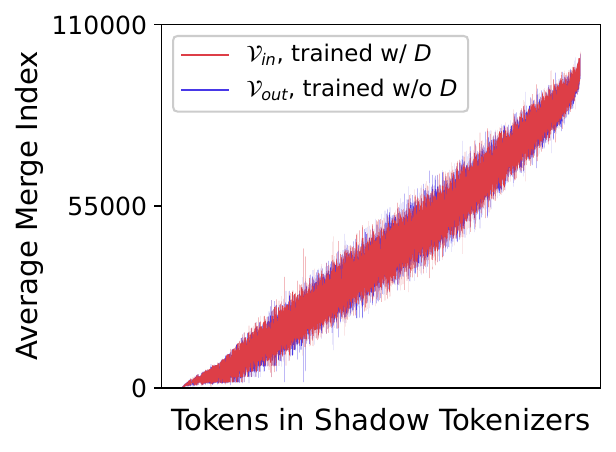}}
    \\
     \subfigure[Vocabulary Size: $140\text{,\,}000$]{\includegraphics[width=0.49\columnwidth]{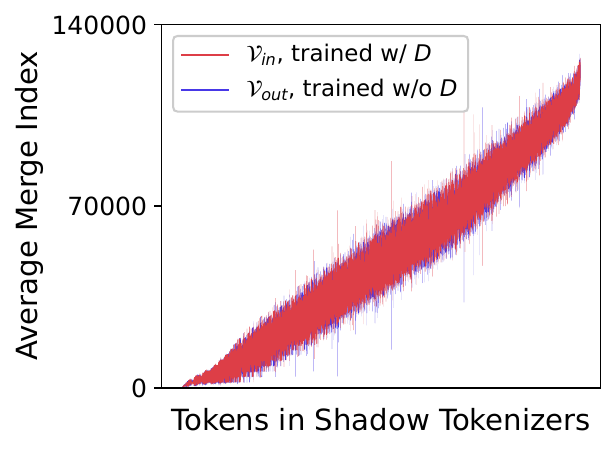}}
    \hspace{0.1em}\subfigure[Vocabulary Size: $170\text{,\,}000$]{\includegraphics[width=0.49\columnwidth]{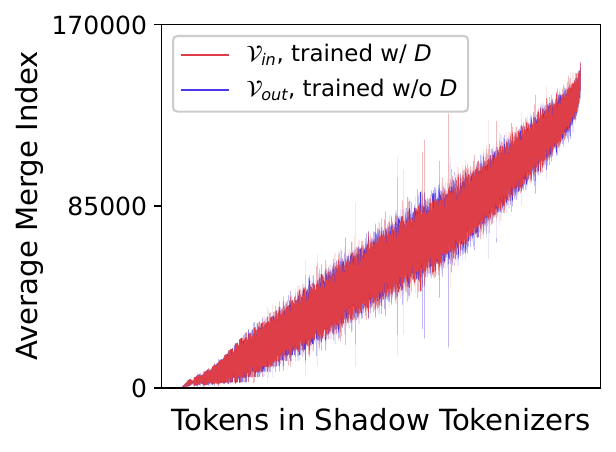}}
\caption{Average merge index for tokens in $\mathcal{V}_\text{in}$ and $\mathcal{V}_\text{out}$. It is shown that overall merge orders in $\mathcal{V}_\text{in}$ and $\mathcal{V}_\text{out}$ resemble.}
    \label{fig: initial}
\end{figure}
We conduct validation experiments for this attack using real-world Internet data~\cite{raffel2020exploring} (see Section~\ref{main exp} for detailed experiments). However, the results demonstrate unsatisfactory
performance of MIA via \msignal in distinguishing between \textit{members} and \textit{non-members}. These results are probably due to the overall distributions of token merge orders in $\mathcal{V}_\text{in}$ and $\mathcal{V}_\text{out}$ resembling each other, as illustrated in Figure~\ref{fig: initial}. The minor discrepancies observed between these distributions suggest weak overfitting signals from the perspective of global tokens. Consequently, the correlation values of Spearman’s rank $\rho(\mathcal{V}_\text{in},\mathcal{V}_\text{target})$ and $\rho(\mathcal{V}_\text{out},\mathcal{V}_\text{target})$ remain highly similar, making it hard for to infer the membership of the dataset $D$.

%% file: tex/3.2_MIA_via_Vocabulary_Overlap.tex
\subsection{\texorpdfstring{Improved MIA via \vsignal}{Improved MIA via Vocabulary Overlap}}\label{improved MIA}
Building on the observation that global token distribution can obscure overfitting signals from the target dataset $D$, we shift our focus to a more fine-grained analysis. Specifically, we solely examine those distinctive tokens whose merge index differs between the vocabularies $\mathbb{V}_{\text{in}}$ and $\mathbb{V}_{\text{out}}$. Our analysis suggests that only when the tokenizer is trained on dataset $D$, some distinctive tokens in $D$ are more likely to be overfit in its vocabulary. Typically, these distinctive tokens more frequently appear in $\mathbb{V}_{\text{in}}$, but are seldom found in $\mathbb{V}_{\text{out}}$. As a result, there exist minor discrepancies between the vocabularies $\mathbb{V}_{\text{in}}$ and $\mathbb{V}_{\text{out}}$. Leveraging this insight, we propose an improved approach, MIA via \vsignal.

\mypara{Design Intuition} When a target tokenizer $f_{\mathcal{V}_\text{target}}$ is trained on a target dataset $D$, its vocabulary $\mathcal{V}_\text{target}$ is likely to overfit the distinctive tokens present in dataset $D$. In fact, existing analysis has shown that OpenAI's tokenizer contains tokens unique to the Reddit forum~\cite{hn36242914,reddittop}. Building on this, we hypothesize that: the more distinctive tokens from $D$ that are found in $\mathcal{V}_\text{target}$, the more likely it is that $\mathcal{V}_\text{target}$ was trained on $D$. To quantify the overlap of distinctive tokens, one effective approach is to use the Jaccard index, which measures the similarity between two sets by focusing on the presence of shared elements. 
Specifically, an adversary can exploit this by computing the \vsignal using Jaccard index $J$ for pairs $(\mathcal{V}_\text{in},\mathcal{V}_\text{target})$ and $(\mathcal{V}_\text{out},\mathcal{V}_\text{target})$ in terms of these distinctive tokens. We write $\mathcal{V}_{\text{non}}=(\bigcup_{\mathcal{V}_\text{in}\in\mathbb{V}_{\text{in}}} \mathcal{V}_\text{in}) \bigcap (\bigcup_{\mathcal{V}_\text{in}\in\mathbb{V}_{\text{out}}} \mathcal{V}_\text{out})$ to denote the set of non-distinctive tokens.
If the average value for $J(\mathcal{V}_\text{in}\backslash\mathcal{V}_{\text{non}}, \mathcal{V}_\text{target}\backslash\mathcal{V}_{\text{non}})$ is higher than that for $J(\mathcal{V}_\text{out}\backslash\mathcal{V}_{\text{non}}, \mathcal{V}_\text{target}\backslash\mathcal{V}_{\text{non}})$, it is more likely $D\in\mathcal{D}_\textup{mem}$. 

\mypara{Attack Method} We structure this attack in five steps. 
\input{algorithm/mia_via_signal_v}
\begin{enumerate}[label={(\roman*)}, leftmargin=1.8em, labelsep=0.3em]
    \item The adversary randomly samples a collection of datasets $\mathcal{D}_\textup{aux}\leftarrow \mathbb{D}$ for $N$ times, and trains $N$ shadow tokenizers considering inclusion or exclusion of the target dataset $D$. This process results in vocabulary sets $\mathbb{V}_{\text{in}}$ and $\mathbb{V}_{\text{out}}$.
    \item The adversary computes the non-distinctive tokens as:
    \begin{equation}\label{non-dis}
        \mathcal{V}_{\text{non}}=(\bigcup_{\mathcal{V}_\text{in}\in\mathbb{V}_{\text{in}}}\hspace{-0.5em} \mathcal{V}_\text{in})\hspace{0.1em} \bigcap \hspace{0.1em}(\bigcup_{\mathcal{V}_\text{in}\in\mathbb{V}_{\text{out}}}\hspace{-0.5em} \mathcal{V}_\text{out}).
    \end{equation}
    \item  The adversary calculates the overfitting signals using the Jaccard index~\cite{bag2019efficient} for each $J(\mathcal{V}_\text{in}\backslash\mathcal{V}_{\text{non}}, \mathcal{V}_\text{target}\backslash\mathcal{V}_{\text{non}})$ and $J(\mathcal{V}_\text{out}\backslash\mathcal{V}_{\text{non}}, \mathcal{V}_\text{target}\backslash\mathcal{V}_{\text{non}})$.
    \item The membership signal for dataset $D$ is defined as:
    \begin{align} 
      \hspace{-2.2em}  \frac{1}{2} &+ \frac{ \sum_{\mathcal{V}_\text{in}\in\mathbb{V}_{\text{in}}}{J(\mathcal{V}_\text{in}\backslash\mathcal{V}_{\text{non}}, \mathcal{V}_\text{target}\backslash\mathcal{V}_{\text{non}})}}{2|\mathbb{V}_{\text{in}}|} \notag \\ &- \frac{ \sum_{\mathcal{V}_\text{out}\in\mathbb{V}_{\text{out}}}J(\mathcal{V}_\text{out}\backslash\mathcal{V}_{\text{non}}, \mathcal{V}_\text{target}\backslash\mathcal{V}_{\text{non}})}{2|\mathbb{V}_{\text{out}}|},
    \end{align}
    where it ranges from 0 to 1. 
    \item If the membership signal is larger than a decision-making threshold $\tau$, output 1 (\textit{member}). Otherwise, output 0.
\end{enumerate}

The detailed process of this attack is outlined in Algorithm~\ref{alg: signal v}. However, like other shadow-based MIAs~\cite{shokri2017membership,ye2022enhanced,carlini2022membership}, we find that MIA via \vsignal requires multiple shadow tokenizers (\eg 96) to effectively capture membership signals. As a result, training such a large number of shadow tokenizers incurs a substantial time cost.

%% file: algorithm/mia_via_signal_v.tex
\begin{algorithm}[t]
\begin{algorithmic}[1]
  \REQUIRE Target dataset $D$, vocabulary of target tokenizer $\mathcal{V}_\text{target}$, underlying distribution $\mathbb{D}$, number of shadow tokenizers $N$, BPE algorithm $\mathcal{T}$, threshold $\tau$
  \STATE $\mathbb{V}_\text{in} \gets \{\}$, ~ $\mathbb{V}_\text{out} \gets \{\}$ 
  \STATE \textcolor{gray!61}{\textit{\# Step 1: Train $N$ shadow tokenizers}}
  \FOR{$N$ times}
    \STATE $\mathcal{D}_{\text{aux}} \gets^\$ \mathbb{D}$ \algcomment{randomly sample auxiliary datasets}
    \STATE $\mathcal{V}_\text{in} \gets \mathcal{T}(\mathcal{D}_{\text{aux}}\bigcup\,\{D\})$ \algcomment{train IN tokenizer}
    \STATE $\mathbb{V}_\text{in} \gets \mathbb{V}_\text{in} \cup \{\mathcal{V}_\text{in}\}$
    \STATE $\mathcal{V}_\text{out} \gets \mathcal{T}(\mathcal{D}_{\text{aux}}\setminus \{D\})$ \algcomment{train OUT tokenizer}
    \STATE $\mathbb{V}_\text{out} \gets \mathbb{V}_\text{out} \cup \{\mathcal{V}_\text{out}\}$
  \ENDFOR
 \STATE \textcolor{gray!61}{\textit{\# Step 2: Compute non-distinctive tokens}}
 \STATE $\mathcal{V}_{\text{non}} \gets \left(\bigcup_{\mathcal{V}_\text{in} \in \mathbb{V}_\text{in}} \mathcal{V}_\text{in} \right) \cap \left( \bigcup_{\mathcal{V}_\text{out} \in \mathbb{V}_\text{out}} \mathcal{V}_\text{out} \right)$ 
 \STATE $J_\text{in} \gets 0$, ~ $J_\text{out} \gets 0$, ~ $\tilde{\mathcal{V}}_\text{target} \gets \mathcal{V}_\text{target} \setminus \mathcal{V}_{\text{non}}$
 \STATE \textcolor{gray!61}{\textit{\# Step 3: Calculate Jaccard index}}
  \FOR{each $\mathcal{V}_\text{in} \in \mathbb{V}_\text{in}$} 
    \STATE $\tilde{\mathcal{V}}_\text{in} \gets \mathcal{V}_\text{in} \setminus \mathcal{V}_{\text{non}}$ \algcomment{filter non-distinctive tokens in $\mathcal{V}_\text{in}$}
    \STATE $J_\text{in} \gets J_\text{in} + \dfrac{|\tilde{\mathcal{V}}_\text{in} \cap \tilde{\mathcal{V}}_\text{target}|}{|\tilde{\mathcal{V}}_\text{in} \cup \tilde{\mathcal{V}}_\text{target}|}$ \algcomment{sum Jaccard index in $\mathbb{V}_\text{in}$}
  \ENDFOR
  \FOR{each $\mathcal{V}_\text{out} \in \mathbb{V}_\text{out}$} 
    \STATE $\tilde{\mathcal{V}}_\text{out} \gets \mathcal{V}_\text{out} \setminus \mathcal{V}_{\text{non}}$ \algcomment{filter non-distinctive tokens in $\mathcal{V}_\text{out}$}
    \STATE $J_\text{out} \gets J_\text{out} + \dfrac{|\tilde{\mathcal{V}}_\text{out} \cap \tilde{\mathcal{V}}_\text{target}|}{|\tilde{\mathcal{V}}_\text{out} \cup \tilde{\mathcal{V}}_\text{target}|}$ \algcomment{sum Jaccard index in $\mathbb{V}_\text{out}$}
  \ENDFOR
  \STATE \textcolor{gray!61}{\textit{\# Step 4: Compute membership signal}}
  \STATE $\textsc{Signal} \gets \dfrac{1}{2} + \dfrac{J_\text{in}}{2|\mathbb{V}_\text{in}|} - \dfrac{J_\text{out}}{2|\mathbb{V}_\text{out}|} $
  \STATE \textcolor{gray!61}{\textit{\# Step 5: Infer the membership}}
  \RETURN $\mathbbm{1}\left[\textsc{Signal} > \tau\right]$ 
\end{algorithmic}
\caption{\textbf{MIA via \vsignal.} 
We train shadow tokenizers with and without the target dataset $D$, filter out non-distinctive tokens, and compute the membership signal. If the signal is larger than a threshold $\tau$, the dataset $D$ is inferred as a member. Otherwise, it is inferred as a non-member.
}
\label{alg: signal v}
\end{algorithm}

%% file: tex/3.3_MIA_via_Frequency_Estimation.tex
\subsection{\texorpdfstring{Efficient MIA via \fsignal}{Efficient MIA via Frequency Estimation}}\label{sec: signal f}
\begin{figure}[t]
    \centering
    \hspace{-0.4em}\subfigure[Vocabulary Size: $80\text{,\,}000$]{\includegraphics[width=0.49\columnwidth]{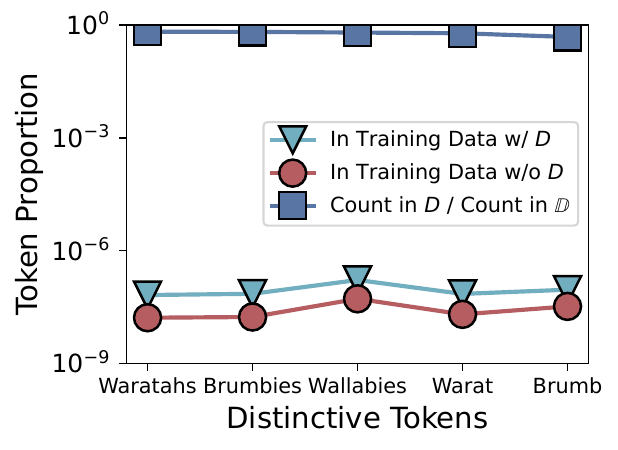}}
    \hspace{0.1em}
    \subfigure[Vocabulary Size: $140\text{,\,}000$]{\includegraphics[width=0.49\columnwidth]{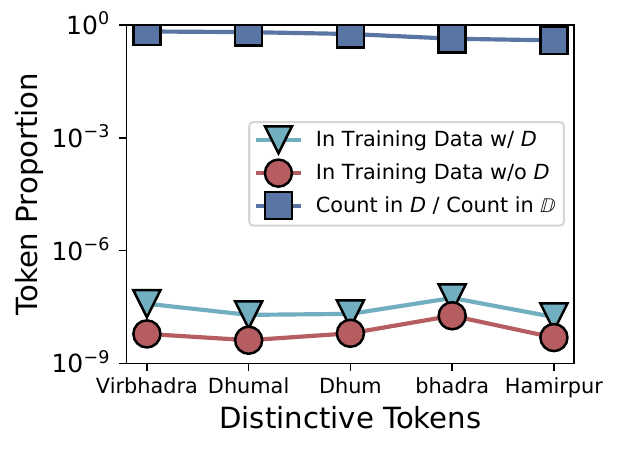}}
\caption{Distinctive tokens in MIA via \vsignal.}
    \label{fig: diff}
\end{figure}
MIA via \vsignal raises a natural question: Can we design an attack that relies on fewer shadow tokenizers and thus reduces the overall time cost? To address this, we investigate whether it is possible to identify distinctive tokens directly by analyzing their statistical characteristics. Motivated by this, we examine such distinctive tokens and derive two key insights: \ding{182} The distinctive tokens of dataset $D$ appear infrequently in the training data of the tokenizer trained on $D$. \ding{183} As shown in Figure~\ref{fig: diff} and Figure~\ref{fig: full diff}, the majority of occurrences of these distinctive tokens in the underlying distribution $\mathbb{D}$ are found in the dataset $D$. 
Given these characteristics, if dataset $D$ is excluded from the tokenizer’s training data, the frequency of such distinctive tokens becomes lower. As a result, these tokens with low frequency are unlikely to be merged into the vocabulary during tokenizer training, since BPE primarily merges the most frequent tokens.
This observation suggests that including dataset $D$ in the training data is almost a necessary condition for some tokens to be merged into the tokenizer's vocabulary. Motivated by these insights, we introduce the MIA via \fsignal.

\mypara{Design Intuition} It is hypothesized that tokenizer training probably exhibits overfitting by incorporating distinctive tokens from the training datasets into its vocabulary~\cite{hn36242914,reddittop}. Building on this intuition, an adversary could exploit such overfitting by evaluating whether including dataset $D$ in the training data is necessary for the merging of some tokens in vocabulary $\mathcal{V}_\text{target}$. If the presence of such distinctive tokens in $\mathcal{V}_\text{target}$ strongly depends on dataset $D$, it is likely $D\in\mathcal{D}_\textup{mem}$.

\mypara{Necessity Evaluation} However, no existing metric evaluates this necessity. To fill this gap, we introduce a new metric: Relative Token Frequency with Self-information (RTF-SI).

\begin{definition}[\textbf{Relative Token Frequency with Self-information}]
Let $\mathbb{D}$ denote a data distribution. Given a dataset $D\subseteq \mathbb{D}$ and a target tokenizer's vocabulary $\mathcal{V}_\text{target}$, the Relative Token Frequency with Self-Information (RTF-SI) of a token $t_i\in \mathcal{V}_\text{target}$ in $D$ is defined as:
\begin{equation}
    \mathrm{RTF\text{-}SI}(D, t_i,  \mathcal{V}_\text{target}) := \operatorname{RTF}(t_i, D) \cdot \operatorname{SI}(t_i, \mathcal{V}_\text{target}),
\end{equation}
where the relative token frequency (RTF) is calculated as:
\begin{equation}
    \operatorname{RTF}(t_i,D) = \frac{n_D(t_i)}{\sum_{D' \in \mathbb{D}} n_{D'}(t_i)},
\end{equation}
with $n_D(t_i)$ denoting the count of token $t_i$ in the dataset $D$. The self-information (SI) is given by:
\begin{equation}
    \operatorname{SI}(t_i, \mathcal{V}_\text{target}) = -\log \Pr\left(t_i \mid \mathcal{V}_\text{target}\right),
\end{equation}
where $\Pr\left(t_i \mid \mathcal{V}_\text{target}\right)$ is the frequency of token $t_i$ appearing in the training data $\mathcal{D}_\textup{mem}$ associated with the vocabulary $\mathcal{V}_\text{target}$. Ideally, this probability is computed as:
\begin{equation}
    \Pr\left(t_i \mid \mathcal{V}_\text{target}\right) = \frac{\sum_{D' \in \mathcal{D}_\textup{mem}} n_{D'}(t_i)}{\sum_{t' \in \mathcal{V}_\text{target}} \sum_{D' \in \mathcal{D}_\textup{mem}} n_{D'}(t')}.
\end{equation}
\end{definition}

\begin{figure}[t]
    \centering
    \hspace{-0.4em}\subfigure[Vocabulary Size: $80\text{,\,}000$]{\includegraphics[width=0.49\columnwidth]{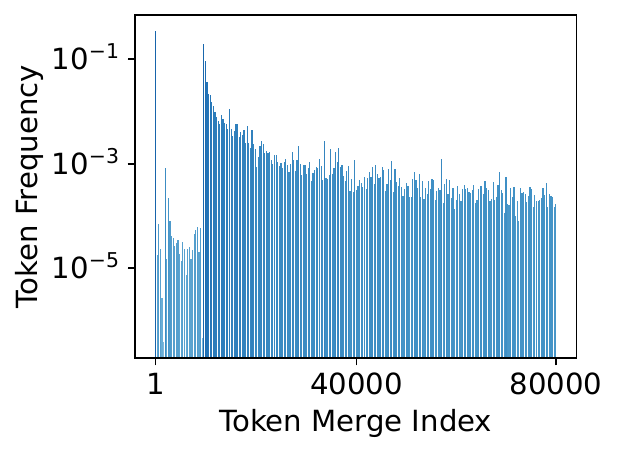}}
    \hspace{0.1em}
    \subfigure[Vocabulary Size: $110\text{,\,}000$]{\includegraphics[width=0.49\columnwidth]{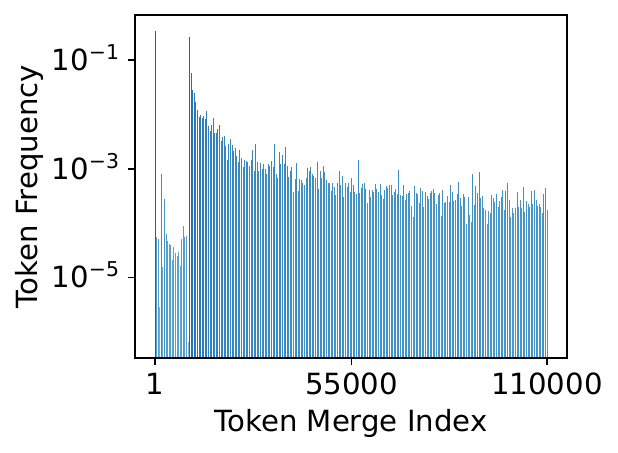}}\\
    \hspace{-0.4em}\subfigure[Vocabulary Size: $140\text{,\,}000$]{\includegraphics[width=0.49\columnwidth]{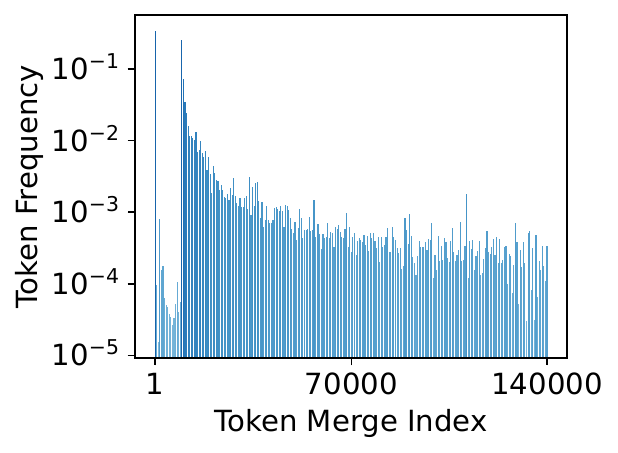}}
    \hspace{0.1em}
    \subfigure[Vocabulary Size: $170\text{,\,}000$]{\includegraphics[width=0.49\columnwidth]{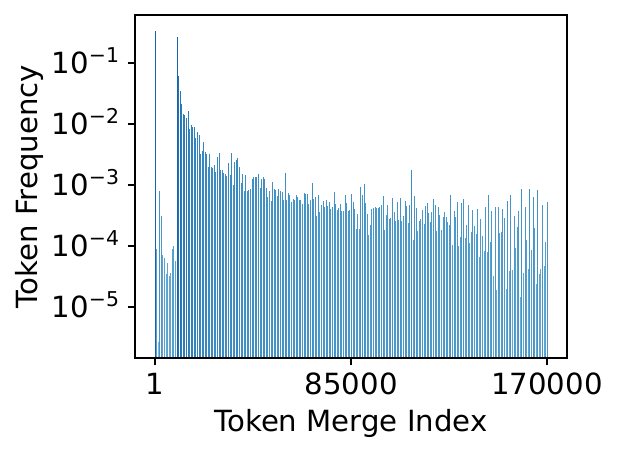}}
\caption{Relationship between token merge index and frequency in training data, indicating they follow a power law.}
    \label{fig: merge}
\end{figure}

RTF-SI evaluates whether it is necessary to include dataset $D$ in the training data for constructing the target vocabulary, $\mathcal{V}_\text{target}$. Building on the classic TF-IDF definition~\cite{ramos2003using,aizawa2003information,sparck1972statistical}, RTF-SI modifies the normalization used in the TF component. A high RTF-SI for the target dataset $D$ suggests that:
\ding{182} Token $t_i$ is with high self-information, \ie it appears infrequently in the training data of the target tokenizer.
\ding{183} Dataset $D$ contributes the majority of token $t_i$'s relative frequency in the underlying distribution $\mathbb{D}$.
As a result, the absence of dataset $D$ may significantly affect some merged tokens in vocabulary $\mathcal{V}_\text{target}$, indicating that $D\in\mathcal{D}_\textup{mem}$.

\input{table/Fit_Power_Law}

\mypara{Frequency Estimation}
In practical implementation, an adversary can estimate the RTF component using auxiliary datasets $\mathcal{D}_{\text{aux}}\gets\mathbb{D}$. However, the SI component is not directly observable, as the frequency $\Pr\left(t_i \mid \mathcal{V}_\text{target}\right)$ in the training data is not available. To estimate this, we draw on the power law~\cite{clauset2009power}, which has been widely applied to approximate word frequency. According to the power law, when a list of measured values exceeding a threshold $x_\textup{min}\in \mathbb{Z}_{> 0}$ is sorted in decreasing order, the $n$-th value is approximately proportional to $1/n^\alpha$, where $\alpha\in\mathbb{R}_{> 0}$ is a constant. As shown in Figure~\ref{fig: merge}, there is a power-law relationship between token merge order and frequency. For rigorous verification, we fit the frequency $\Pr\left(t_i \mid \mathcal{V}_\text{target}\right)$ in a power-law distribution~\cite{clauset2009power}:
\begin{equation}\label{eq: pw}
\Pr(t_i \mid \mathcal{V}_\text{target}) \propto \frac{1}{i^\alpha},~\text{where}~t_i \in \mathcal{V}_\text{target}~\text{and}~i > x_\textup{min}.
\end{equation}
The estimation results in Table~\ref{fit powerlaw} show a small standard error between the estimated and actual values, which supports using the power law to approximate the value $\Pr\left(t_i \mid \mathcal{V}\right)$ in the SI component. Given this, RTF-SI can also be computed as:
\input{algorithm/mia_via_signal_f}

\begin{theorem}[\textbf{RTF-SI under the Power Law}]\label{theo}
Under the power-law distribution~\textup{\cite{clauset2009power}}, the frequency $\Pr(t_i \mid \mathcal{V}_\text{target})$ of a token $t_i \in \mathcal{V}_\text{target}$ is proportional to $1/i^\alpha$:
\begin{equation}
\Pr(t_i \mid \mathcal{V}_\text{target}) \propto \frac{1}{i^\alpha},
\end{equation}
where $i > x_\textup{min}$, and $\alpha\in\mathbb{R}_{> 0}$\,, $x_\textup{min}\in\mathbb{Z}_{> 0}$ are constants defined by the power law.
Then, RTF-SI can be approximated by its lower bound:
\begin{equation}
\operatorname{RTF\text{-}SI}( D, t_i, \mathcal{V}_\text{target}) \geq \frac{n_D(t_i)}{\sum_{D' \in \mathbb{D}} n_{D'}(t_i)} \cdot \log (\hspace{-0.9em}\sum_{j=x_\textup{min}+1}^{|\mathcal{V}_\text{target}|}\hspace{-0.9em}{i^\alpha}/{j^\alpha}).
\end{equation}
\end{theorem}

The detailed proof of Theorem~\ref{theo} is provided in Appendix~\ref{a: proof}. The power law allows an adversary to estimate RTF-SI without directly accessing the frequency $\Pr\left(t_i \mid \mathcal{V}_\text{target}\right)$. Since the power law estimates the frequency of tokens $t_i\in\mathcal{V}_\text{target}$ with merge index $i > x_{\text{min}}$, MIA via \fsignal also concentrates on them.

\mypara{Attack Method} We outline this attack in four steps below.

\begin{enumerate}[label={(\roman*)}, leftmargin=1.8em, labelsep=0.3em]
    \item The adversary randomly samples a collection of datasets $\mathcal{D}_\textup{aux}\leftarrow \mathbb{D}$ $N$ times, comprising a set $\tilde{\mathbb{D}}$. Then, the adversary trains a shadow tokenizer $f_{\mathcal{V}_\text{shadow}}$ using a $\mathcal{D}_\textup{aux} \subseteq \tilde{\mathbb{D}}$. 
    \item The adversary fits the power-law distribution in Equation~\ref{eq: pw} using the vocabulary $\mathcal{V}_\text{shadow}$ and its training data. For each token $t_i\in\mathcal{V}_\text{target}$ where $i > x_{\text{min}}$, the adversary approximate its $\operatorname{RTF}(D, t_i)$ on set $\tilde{\mathbb{D}}\bigcup\{D\}$, and estimate its $\operatorname{SI}(t_i, \mathcal{V}_\text{target})$ via the fitted power-law distribution.
    \item Based on the Theorem~\ref{theo}, the membership signal for target dataset $D$ is defined as follows:
    \begin{equation}\label{eq: signal f}
    \hspace{-2.0 em}  \sigma\big( \max_{t_i\in\mathcal{V}_\text{target},\,i > x_{\text{min}}} \underbrace{\frac{n_D(t_i)}{\sum_{D' \in \tilde{\mathbb{D}}\bigcup\{D\}} n_{D'}(t_i)} }_{\operatorname{RTF}(D, t_i)}\cdot \underbrace{\log (\hspace{-0.9em}\sum_{j=x_\textup{min}+1}^{|\mathcal{V}_\text{target}|}\hspace{-0.9em}{i^\alpha}/{j^\alpha})}_{\operatorname{SI}(t_i, \mathcal{V}_\text{target})}\hspace{-0.2em}\big),
    \end{equation}
    where $\sigma$ denotes the sigmoid function~\cite{rumelhart1986learning}. Thereby, the value of Equation~\ref{eq: signal f} ranges from 0 to 1.
    \item If the membership signal is larger than a decision-making threshold $\tau$, output 1 (\textit{member}). Otherwise, output 0.
\end{enumerate}

The detailed process of this attack is shown in Algorithm~\ref{alg: signal f}. The membership signal for MIA via \fsignal is defined as the maximum RTF-SI value for the target dataset $D$. Specifically, if including dataset $D$ in the training data is necessary for at least one token $t_i$ to be merged into the vocabulary $\mathcal{V}_\text{target}$, it suggests that the absence of dataset $D$ has a significant influence on the already constructed vocabulary. Consequently, it is likely that dataset $D$ is a \textit{member}.

%% file: table/Fit_Power_Law.tex
\begin{table}[t!]
    \centering
    \caption{Power-law fit on token frequency in training data.}
    \resizebox{0.9\columnwidth}{!}{
    \begin{tabular}{l|ccccc}
    \toprule
    \rule{0pt}{2.5ex} $|\mathcal{V}_\text{target}|$ & 80{,\,}000 & 110{,\,}000 & 140{,\,}000 & 170{,\,}000 & 200{,\,}000 \\
    \midrule
    $x_\textup{min}$          & 9{,\,}782 &  9{,\,}782 &  9{,\,}782 &  9{,\,}782 &  9{,\,}782 \\
    $\alpha$                  & 3.831 & 3.409 & 3.111 & 2.914 & 2.764 \\
    Std. Error                  & 0.001 & 0.001 & 0.002 & 0.001 & 0.001 \\
    \bottomrule
    \end{tabular}}
    \label{fit powerlaw}
\end{table}

%% file: algorithm/mia_via_signal_f.tex
\begin{algorithm}[t!]
\begin{algorithmic}[1]
  \REQUIRE Target dataset $D$, vocabulary of target tokenizer $\mathcal{V}_\text{target}$, underlying distribution $\mathbb{D}$, sampling times $N$, BPE algorithm $\mathcal{T}$, power-law fit function \texttt{pl.fit}, threshold $\tau$
  \STATE $\tilde{\mathbb{D}} \gets \{\}$
  \STATE \textcolor{gray!61}{\textit{\# Step 1: Prepare for frequency estimation}}
  \FOR{$N$ times}
    \STATE $\mathcal{D}_{\text{aux}} \gets^\$ \mathbb{D}$ \algcomment{randomly sample auxiliary datasets}
    \STATE $\tilde{\mathbb{D}} \gets \tilde{\mathbb{D}} \cup \mathcal{D}_{\text{aux}}$
  \ENDFOR
 \STATE $\mathcal{V}_\text{shadow}\gets\mathcal{T}(\mathcal{D}_{\text{aux}})$\algcomment{train shadow tokenizer}
 \STATE \textcolor{gray!61}{\textit{\# Step 2: Estimate components in RTF-SI}}
 \STATE $\alpha,\,x_\textup{min}\gets\text{\texttt{pl.fit}}(\mathcal{V}_\text{shadow},\mathcal{D}_{\text{aux}})$ \algcomment{fit token frequency}
 \FOR{$i = x_\textup{min}+1$ to $|\mathcal{V}_\text{target}|$}
    \STATE $\operatorname{RTF}(D, t_i)\gets\dfrac{n_D(t_i)}{\sum_{D' \in \tilde{\mathbb{D}}\bigcup\{D\}} n_{D'}(t_i)}$
    \STATE $\operatorname{SI}(t_i, \mathcal{V}_\text{target}) \gets \displaystyle\log(\hspace{-0.5em}\sum_{j=x_\textup{min}+1}^{|\mathcal{V}_\text{target}|} \frac{i^\alpha}{j^\alpha})$ \algcomment{apply Theorem 4.2}
 \ENDFOR
 
 \STATE \textcolor{gray!61}{\textit{\# Step 3: Compute membership signal}}
 \FOR{$i = x_\textup{min}+1$ to $|\mathcal{V}_\text{target}|$}
    \STATE $\mathrm{RTF\text{-}SI}(D, t_i,  \mathcal{V}_\text{target}) \gets \operatorname{RTF}(t_i, D) \cdot \operatorname{SI}(t_i, \mathcal{V}_\text{target})$
\ENDFOR
\STATE $\mathrm{RTF\text{-}SI}_{\max}\gets \max_{t_i\in\mathcal{V}_\text{target},\,i > x_{\text{min}}} \mathrm{RTF\text{-}SI}(D, t_i,  \mathcal{V}_\text{target})$ 
  \STATE $\textsc{Signal} \gets \dfrac{1}{1 + e^{-{\mathrm{RTF\text{-}SI}_{\max}}}}$ \algcomment{normalize by sigmoid}
  \STATE \textcolor{gray!61}{\textit{\# Step 4: Infer the membership}}
  \RETURN $\mathbbm{1}\left[\textsc{Signal} > \tau\right]$ 
\end{algorithmic}
\caption{\textbf{MIA via \fsignal.} 
We train a shadow tokenizer to fit the power-law distribution of token frequency, approximate RTF-SI for each token $t_i\in\mathcal{V}_\text{target}$ where $i > x_{\text{min}}$, and compute the membership signal based on the maximum RTF-SI. If the membership signal is larger than a decision-making threshold $\tau$, the dataset $D$ is inferred as a member. Otherwise, it is inferred as a non-member.
}
\label{alg: signal f}
\end{algorithm}

%% file: tex/4_Experiments.tex
\section{Attack Evaluation}\label{mia_evaluation}
In this section, we first introduce the experimental settings. Next, we develop two shadow-free membership inference methods to serve as additional exploration and baselines. Finally, we present the evaluation results.
\subsection{Experimental Setup}\label{setup}
\mypara{Datasets} According to disclosures from existing LLMs~\cite{touvron2023llama,biderman2023pythia,guo2025deepseek}, the training data for tokenizers is primarily sourced from publicly available web content. To ensure a realistic evaluation of our attacks, we therefore utilize real-world web data from the \textit{C4} corpus~\cite{raffel2020exploring} in our evaluation. Specifically, the \textit{C4} corpus is created by cleaning and filtering web pages from Common Crawl, widely used for training and evaluating natural language processing models. In our evaluation, we utilize 1,681,296 web pages across 4,133 websites (i.e., $\mathbb{D}$) from the \textit{C4} corpus, with each website treated as a dataset $D$. 

\begin{figure*}[t]
\centering
\includegraphics[width=1\textwidth]{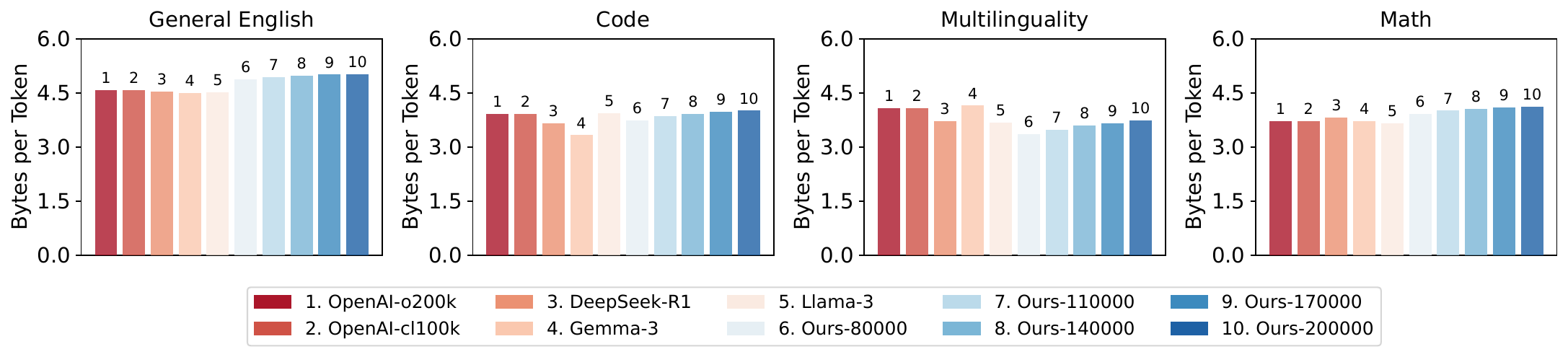}
\caption{Comparison of tokenizer utility based on the metric of bytes per token. Specifically, “Ours-80000” refers to our trained target tokenizer with a vocabulary size of 80,000 tokens. The above experimental results indicate that the utility performance of the target tokenizers utilized in our evaluations is comparable to that of tokenizers used in state-of-the-art LLMs~\cite{openai2025tiktoken, karter2024gemini, guo2025deepseek}.}
\label{fig: cmp tokenizers}
\end{figure*}
\begin{figure}[t]
\centering
\includegraphics[width=0.8\columnwidth]{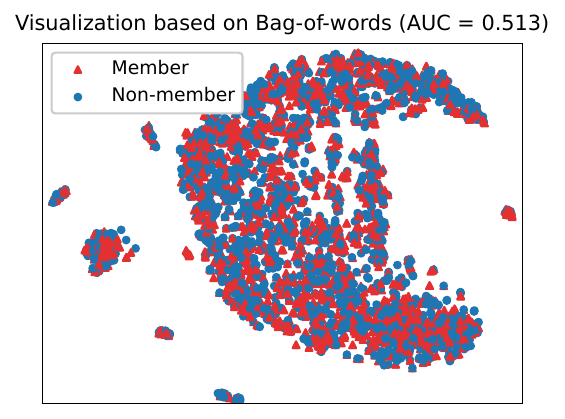}
\caption{Visualization of the test set using t-SNE~\cite{maaten2008visualizing}. The results confirm no distribution shifts in our evaluation.}
\label{fig: blind mia}
\end{figure}

\mypara{Tokenizer Training} Following prior work~\cite{carlini2022membership,shokri2017membership}, we randomly select half of the datasets in $\mathbb{D}$ to serve as training data. Consistent with DeepSeek~\cite{bi2024deepseek}, we train the target tokenizers using the HuggingFace library. The tokenizer vocabulary size ranges from 80,000 to 200,000 tokens, with the upper bound matching that of OpenAI’s latest tokenizer, o200k~\cite{openai2025tiktoken}.
For MIA via \msignal and MIA via \vsignal, the adversary trains 96 shadow tokenizers. For MIA via \fsignal, the adversary samples the auxiliary datasets 10 times to compose a set $\tilde{\mathbb{D}}$.

\mypara{Verification of No Distribution Shifts} As discussed in previous works~\cite{das2025blind,duanmembership,meeus2024sok}, distribution shifts between \textit{members} and \textit{non-members} can invalidate the evaluation of MIAs.
In such scenarios, an evaluator~\cite{das2025blind,meeus2024sok} can exploit bag-of-words features extracted from test samples and train a random forest classifier to detect whether a sample was part of the training data, even without access to the target model.
To ensure that there is no distribution shift in our evaluation, we follow the methodology of prior work~\cite{das2025blind} by training a random forest classifier and leveraging the bag-of-words features to distinguish between \textit{members} and \textit{non-members}. The experimental results are illustrated in Figure~\ref{fig: blind mia}.
The AUC score of 0.513 via this approach to distinguish \textit{members} and \textit{non-members} confirms the absence of distribution shifts in our evaluations.

\mypara{Comparison to Commercial Tokenizers} To assess how well our trained target tokenizers in experimental settings mimic the commercial tokenizers in the real world, we compare their utility performance using a standard compression metric: the bytes per token~\cite{liu-etal-2025-superbpe}. Specifically, following prior work~\cite{slagle2024spacebyte,dagan2024getting}, we compute the ratio of UTF-8 bytes in a given text to the number of tokens generated by the tokenizer. A higher score is desired. We conduct this evaluation across widely used benchmarks: general English (\textit{WikiText-103}~\cite{merity2016pointer}), code (\textit{GitHub Code}~\cite{codeparrot_github_code_2025}), multilingual content (\textit{MGSM}~\cite{saparov2023language}), and mathematics (\textit{GPQA}~\cite{rein2024gpqa}). As shown in Figure~\ref{fig: cmp tokenizers}, the utility performance of our trained tokenizers is comparable to that of the commercial tokenizers~\cite{openai2025tiktoken,karter2024gemini,guo2025deepseek,touvron2023llama}. 
\input{table/Main_Results} 
\begin{figure*}[t!]
     \centering   
    \subfigure[Vocabulary Size: 80,000]{\includegraphics[width=0.32\textwidth]{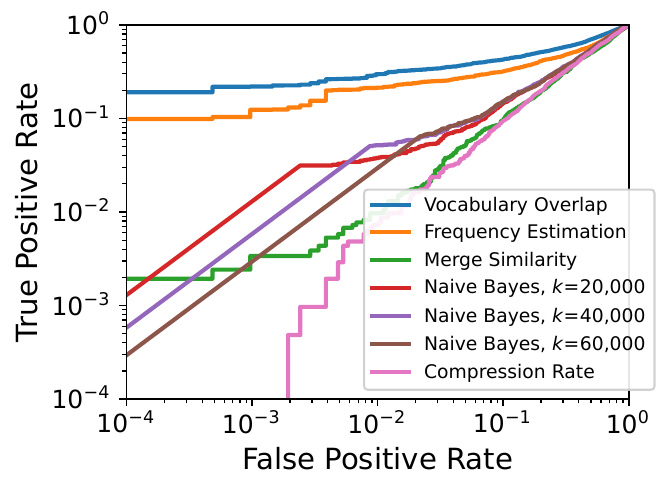}}
    \subfigure[Vocabulary Size: 140,000]{\includegraphics[width=0.32\textwidth]{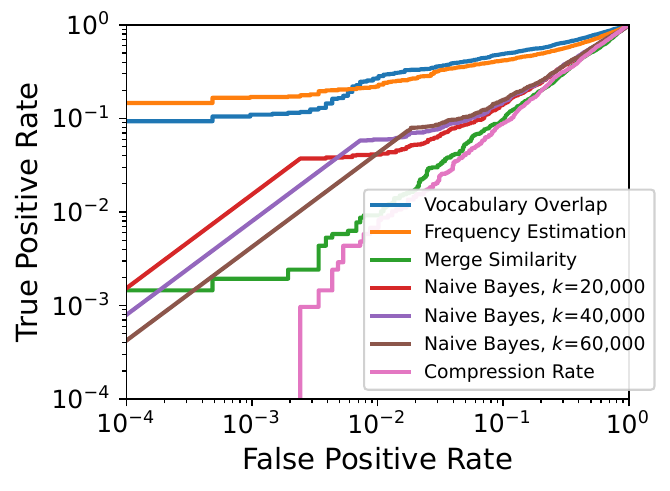}}
        \subfigure[Vocabulary Size: 200,000]{\includegraphics[width=0.32\textwidth]{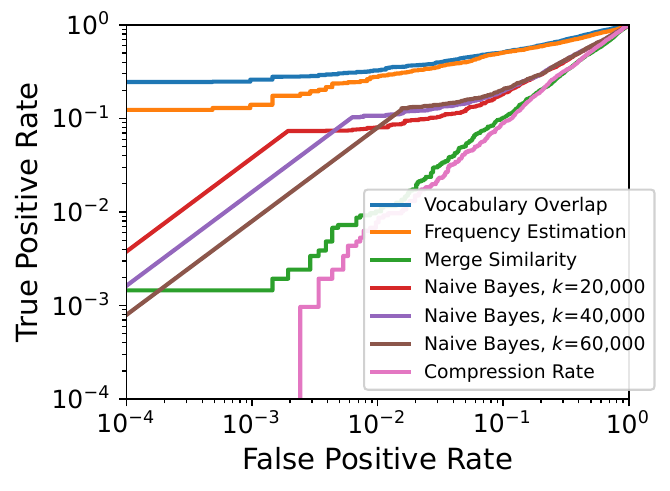}}
\caption{Success rate of our attacks on tokenizers with different vocabulary sizes. The experimental results demonstrate that, when evaluated at low false positive rates, both MIA via \vsignal and MIA via \fsignal consistently outperform other methods. Notably, even at a $0.01\%$ false positive rate, both attacks achieve a true positive rate of nearly 10\%.}
    \label{fig: ROC curves}
\end{figure*}

\mypara{Attack Baselines} As our work is the first to investigate MIAs targeting tokenizers, there are no existing baselines from prior studies. Therefore, we establish our own baselines by comparing MIA via \vsignal\ and MIA via \fsignal\ with three other methods: MIA via \msignal\ (see Section~\ref{initial_exploration}), as well as two attack methods we developed. These two attacks are described below:
\begin{itemize}
\item {\mypara{MIA via \nsignal}}
Since every token originates from at least one of the tokenizer's training datasets, the adversary can approximate the empirical probability $\Pr(t_i \rightarrow D)$ that the token $t_i \in \mathcal{V}_\text{target}$ comes from the dataset $D$. Specifically, this probability $\Pr(t_i \rightarrow D)$ can be represented by:
\begin{equation}
\Pr(t_i \rightarrow D) = \frac{n_D(t_i)}{\sum_{D' \in \tilde{\mathbb{D}}\cup\{D\}} n_{D'}(t_i)},
\end{equation}
where $n_D(t_i)$ is the frequency of $t_i$ in $D$, and $\tilde{\mathbb{D}}$ is constructed by sampling auxiliary datasets $\mathcal{D}_\textup{aux}$ $N$ times as an estimation for distribution ${\mathbb{D}}$. We default to setting $N=10$ in evaluation. The probability that $D \in \mathcal{D}_\textup{mem}$ is given by:
\begin{equation}
\Pr(D \in \mathcal{D}_\textup{mem}) = 1 - \Pr\big( \bigcap_{t_i \in \mathcal{V}_\text{target}} { t_i \not\rightarrow D } \big),
\end{equation} where $t_i \not\rightarrow D$ means $t_i$ was not sourced from $D$. Assuming rare tokens in training data typically come from disjoint datasets, the probability $\Pr(D \in \mathcal{D}_\textup{mem}) $ thus can be approximated via the Naive Bayes~\cite{webb2017naive}:
\begin{align}
\Pr(D \in \mathcal{D}_\textup{mem}) &\approx 1 - \prod_{t_i \in \mathcal{V}_\textup{rare}} \Pr(t_i \not\rightarrow D) \\
&\approx 1 - \prod_{t_i \in \mathcal{V}_\textup{rare}} \big(1 - \Pr(t_i \rightarrow D)\big),
\end{align}
where $\mathcal{V}_\textup{rare} \subseteq \mathcal{V}_\text{target}$ is the set of tokens with the top $k$ merge orders, selected as likely rare tokens in training data (see Equation~\ref{eq: pw}). The membership signal for the target dataset $D$ is defined by the probability $\Pr(D \in \mathcal{D}_\textup{mem})$. If it is larger than a threshold $\tau$, output 1 (\textit{member}); or, output 0.

\item {\mypara{MIA via \csignal}} 
The objective of tokenizer training is to maximize the compression rate of a given text corpus~\cite{zouhar2023formal}. 
Based on this optimization objective, we hypothesize that a tokenizer achieves higher compression rates on datasets it was trained on~\cite{hayasedata}. 
Leveraging this insight, an adversary can calculate the compression rate of a given dataset $D$ using the metric of bytes per token~\cite{liu-etal-2025-superbpe}.
The membership signal for the target dataset $D$ is defined by this compression rate. If it is larger than a decision-making threshold $\tau$, output 1 (\textit{member}); otherwise, output 0.

\end{itemize}

\mypara{Evaluation Metrics} Following prior studies~\cite{shokri2017membership,duanmembership,carlini2022membership}, we report the MIA performance using three convincing metrics:
\begin{itemize} 
\item \mypara{AUC}  This metric quantifies the overall distinguishability of an MIA by computing the area under the receiver operating characteristic (ROC) curve~\cite{carrington2022deep}.
\item \mypara{Balanced Accuracy} This metric (denoted as BA) measures the overall correct predictions on membership by averaging the true positive rate and the true negative rate.
\item \mypara{TPR at Low FPR}
Proposed by~\cite{carlini2022membership}, this metric evaluates the true positive rate (TPR) when the false positive rate (FPR) is low. Following prior work~\cite{meeus2024did,25Towards}, we report TPR @ 1.0\% FPR in our evaluations (denoted as TPR). 
\end{itemize}

\input{tex/4.1_Evaluation_of_MIAs_against_Tokenizers}

\input{tex/4.2_MIA_for_Dataset_with_varying_sizes}

\input{tex/4.3_MIA_against_Tokenizers_with_Min_Count_Defense}

\input{tex/4.4_Additional_Investigations}

%% file: table/Main_Results.tex
\begin{table*}[t]
    \centering
    \caption{Comparison of MIAs against target tokenizers. Here, BA denotes the metric of balanced accuracy, and TPR refers to the metric of TPR @ 1.0\% FPR. The bold values indicate the best performance, while the underlined values denote the second-best. It is observed that MIA via \vsignal and MIA via \fsignal outperform other baseline methods.}

    \resizebox{1\textwidth}{!}{
    \begin{tabular}{@{}l@{\hskip 8pt}
                    c@{\hskip 8pt}
                    c@{\hskip 8pt}
                    ccc ccc ccc ccc ccc@{}}
 
        & \raisebox{1.8ex}{\multirow{3}{*}{\rotatebox{90}{\parbox[c]{1.4cm}{\centering \footnotesize Shadow Tokenizers}}}}
        & \raisebox{1.8ex}{\multirow{3}{*}{\rotatebox{90}{\parbox[c]{1.4cm}{\centering \footnotesize Auxiliary Datasets}}}}
        & \multicolumn{3}{c}{\raisebox{0.3ex}{\textbf{$|\mathcal{V}_\text{target}|=80{,}000$}}} 
        & \multicolumn{3}{c}{\raisebox{0.3ex}{\textbf{$|\mathcal{V}_\text{target}|=110{,}000$}}} 
        & \multicolumn{3}{c}{\raisebox{0.3ex}{\textbf{$|\mathcal{V}_\text{target}|=140{,}000$}}} 
        & \multicolumn{3}{c}{\raisebox{0.3ex}{\textbf{$|\mathcal{V}_\text{target}|=170{,}000$}}}  
        & \multicolumn{3}{c}{\raisebox{0.3ex}{\textbf{$|\mathcal{V}_\text{target}|=200{,}000$}}}  \\
        \cmidrule(l{1pt}r{13pt}){4-6} \cmidrule(l{6pt}r{13pt}){7-9} 
        \cmidrule(l{6pt}r{13pt}){10-12} \cmidrule(l{6pt}r{13pt}){13-15} \cmidrule(l{6pt}r{0pt}){16-18}
       \textbf{Attack Approach} & & 
        & AUC & BA & TPR
        & AUC & BA & TPR 
        & AUC & BA & TPR 
        & AUC & BA & TPR 
        & AUC & BA & TPR \\
        \midrule
    \csignal
    &\Circle &\Circle
    & 0.507 & 0.513 & 0.72\% 
    & 0.509 & 0.517 & 0.71\% 
    & 0.508 & 0.514 & 0.68\%  
    & 0.508 & 0.513 & 0.82\%  
    & 0.509 & 0.516 & 0.77\%  \\
    \makecell{\nsignal, $k\!=\!20{,}000$}
     &\Circle &\CIRCLE 
    & 0.534 & 0.526 & 3.78\% 
    & 0.526 & 0.524 & 3.44\% 
    & 0.535 & 0.533 & 4.11\% 
    & 0.546 & 0.538 & 5.86\%  
    & 0.564 & 0.551 & 8.03\% \\
    \makecell{\nsignal, $k\!=\!40{,}000$}
    &\Circle &\CIRCLE 
    & 0.543 & 0.530 & 5.18\% 
    & 0.546 & 0.533 & 6.10\% 
    & 0.542 & 0.537 & 5.95\% 
    & 0.550 & 0.542 & 7.60\%  
    & 0.572 & 0.557 & 10.70\% \\
    \makecell{\nsignal, $k\!=\!60{,}000$}
    &\Circle &\CIRCLE 
    & 0.543 & 0.530 & 2.22\% 
    & 0.551 & 0.538 & 5.57\% 
    & 0.543 & 0.536 & 2.80\%  
    & 0.553 & 0.545 & 4.94\%  
    & 0.572 & 0.557 & 8.86\%  \\

    \msignal
    &\CIRCLE &\CIRCLE 
    & 0.493 & 0.509 & 1.06\% 
    & 0.494 & 0.507 & 1.02\% 
    & 0.494 & 0.508 & 0.97\%  
    & 0.495 & 0.506 & 0.92\%  
    & 0.495 & 0.508 & 0.87\%  \\
    \vsignal
    &\CIRCLE &\CIRCLE 
    &\textbf{0.693} & \textbf{0.666} & \textbf{26.77\% }
    & \textbf{0.718} & \textbf{0.672} & \textbf{28.75\%} 
    & \textbf{0.736 }& \textbf{0.696} & \textbf{29.72\% } 
    & \textbf{0.761} & \textbf{0.709} & \textbf{32.53\%}  
    & \textbf{0.771} & \textbf{0.711} & \textbf{34.61\%}  \\
    \fsignal
    &\CIRCLE &\CIRCLE 
    & \underline{0.610} & \underline{0.614} & \underline{21.30\%} 
    & \underline{0.645} & \underline{0.641} & \underline{22.00\%} 
    & \underline{0.676} & \underline{0.660} & \underline{22.41\%}  
    & \underline{0.707} & \underline{0.681} & \underline{25.61\% } 
    & \underline{0.740} & \underline{0.705} & \underline{27.88\%}  \\
    \bottomrule
    \end{tabular}
    }
    \label{tab:attack-full-metrics}
\end{table*}

%% file: tex/4.1_Evaluation_of_MIAs_against_Tokenizers.tex
\vspace{-1em}
\subsection{Main Results}\label{main exp}

\begin{tcolorbox}[enhanced, drop shadow southwest, top=2pt,bottom=2pt]
\textbf{Finding 1.} \textit{According to prior work~\textup{\cite{ huangover,mayilvahanan2025llms}}, scaling up the intelligence of LLMs involves expanding the tokenizer's vocabulary~\textup{\cite{tao2024scaling}} and thus improving its compression efficiency~\textup{\cite{liu-etal-2025-superbpe}}. Yet, the results show it also increases a tokenizer's vulnerability to MIAs.}
\end{tcolorbox}

\mypara{Overall Performance}
As shown in Table~\ref{tab:attack-full-metrics}, the MIA via \vsignal and the  MIA via \fsignal consistently demonstrate strong performance and outperform other baseline methods across different vocabulary sizes. For instance, the MIA via \vsignal achieves an AUC score of 0.771 when evaluated on a target tokenizer with 200,000 tokens, whereas the MIA via \fsignal achieves a comparable AUC score of 0.740.
 Beyond these results, we observe that the performance of both attacks improves as vocabulary size increases. This trend suggests that MIAs targeting tokenizers in LLMs may become even more effective as state-of-the-art models continue to scale and adopt larger vocabularies in their tokenizers~\cite{tao2024scaling, huangover, mayilvahanan2025llms}. One possible explanation is that larger vocabularies contain more tokens, which may increase the likelihood of merging the distinctive tokens from the training data. As a result, expanding the tokenizer’s vocabulary may unintentionally increase its vulnerability to effective MIAs. 
\input{table/Dataset_Scores}
\input{table/training_time_for_shadow_tokenizers}


\mypara{ROC Analysis} The prior study~\cite{carlini2022membership} has highlighted the importance of MIAs being able to reliably infer even a small number of a model's training data. To demonstrate this capability of our MIAs, Figure~\ref{fig: ROC curves} presents the full log-scale ROC curves for various attack methods across different vocabulary sizes. It is observed that both MIA via \vsignal and MIA via \fsignal can reliably infer the membership of datasets, particularly in regions with low false positive rates. For example, when applied to a tokenizer with 140,000 tokens, these two attacks achieve true positive rates ranging from approximately 10\% to 30\% at a false positive rate below 1\%. Even at a false positive rate of  $0.01\%$, both attacks can still achieve a true positive rate of nearly 10\%. Additional ROC curve results can be found in Figure~\ref{fig: full ROC curves}.

\mypara{Efficiency Analysis} We further analyze the computational cost for both MIA via \vsignal and MIA via \fsignal across two phases: shadow tokenizer training and the remaining inference. In the phase of shadow tokenizer training, MIA via \vsignal trains multiple tokenizers (e.g., 96), resulting in a high computational cost. In contrast, MIA via \fsignal requires training only a single tokenizer, significantly reducing the overall cost. As shown in Table~\ref{training_time}, this leads to substantial savings in training time. In the inference phase, MIA via \vsignal involves frequent comparisons across different tokenizers, whereas MIA via \fsignal primarily estimates a power-law distribution, a much simpler computation. Table~\ref{MIA efficiency} confirms the shorter inference time of the latter method. For example, MIA via \vsignal takes over two hours to infer the membership of 4,133  datasets from a tokenizer with 140,000 tokens. However, MIA via \fsignal accomplishes the same task in under 20 minutes, making it efficient for large-scale attacks.

%% file: table/Dataset_Scores.tex
\begin{table*}[t]
    \centering
    \caption{Impact of the target dataset size on MIAs. BA denotes balanced accuracy. TPR refers to TPR @ 1.0\% FPR. It is observed that MIA via \vsignal and MIA via \fsignal perform better for target datasets with larger sizes.}

    \resizebox{1\textwidth}{!}{
    \begin{tabular}{@{}l c
                    ccc ccc ccc ccc ccc@{}}
        \toprule
        \multirow{2}{*}{\textbf{Attack Approach}}
        & \multirow{2}{*}{{\#Dataset Size}}
        & \multicolumn{3}{c}{\textbf{$|\mathcal{V}_\text{target}|=80{,}000$}} 
        & \multicolumn{3}{c}{\textbf{$|\mathcal{V}_\text{target}|=110{,}000$}} 
        & \multicolumn{3}{c}{\textbf{$|\mathcal{V}_\text{target}|=140{,}000$}} 
        & \multicolumn{3}{c}{\textbf{$|\mathcal{V}_\text{target}|=170{,}000$}}  
        & \multicolumn{3}{c}{\textbf{$|\mathcal{V}_\text{target}|=200{,}000$}}  \\
        \cmidrule(l{7pt}r{13pt}){3-5} \cmidrule(l{6pt}r{13pt}){6-8} 
        \cmidrule(l{6pt}r{13pt}){9-11} \cmidrule(l{6pt}r{13pt}){12-14} \cmidrule(l{6pt}r{0pt}){15-17}
        & 
        & AUC & BA & TPR
        & AUC & BA & TPR 
        & AUC & BA & TPR 
        & AUC & BA & TPR 
        & AUC & BA & TPR \\
        \midrule
    \multirow{3}{*}{\vsignal}
    & {$|D|\in[0,400)$} 
    & 0.672 & 0.652 & 21.58\% %
    & 0.696 & 0.655 & 24.76\% %
    & 0.714 & 0.676 &  27.56\%  %
    & 0.737 & 0.687 & 27.82\%  %
    & 0.747 & 0.692 &  29.22\%  \\ %
    & {$|D|\in[400,800)$} 
    & 0.739 & 0.695 & 33.73\% %
    & 0.758 & 0.718 & 41.27\% %
    & 0.791 & 0.755 & 42.77\% %
    & 0.808 & 0.761 & 43.37\%  
    & 0.808 & 0.766 & 43.67\%  \\
    & {$|D|\in[800,1200)$} 
    & 0.773 & 0.720 & 33.75\% %
    & 0.785 & 0.757 & 43.75\%
    & 0.797 & 0.767 & 45.00\%  
    & 0.826 & 0.813 & 47.50\%
    & 0.882& 0.838  & 62.50\%  \\
    \midrule
    \multirow{3}{*}{\fsignal}
    & {$|D|\in[0,400)$} 
    & 0.599 & 0.608 & 18.91\%  
    & 0.631 & 0.632 &  19.73\%
    & 0.662 & 0.648 &  20.31\%
    & 0.695 & 0.668 & 21.83\%  
    & 0.729 & 0.691 & 25.84\%  \\
    & {$|D|\in[400,800)$}
    & 0.629 & 0.624 & 23.49\% 
    & 0.683 & 0.662 & 29.27\%
    & 0.728 & 0.697 & 30.42\%  
    & 0.747 & 0.713 & 31.63\%  
    & 0.774 & 0.736 & 32.83\% \\
    & {$|D|\in[800,1200)$} 
    & 0.758 & 0.734 & 33.75\% 
    & 0.772 &  0.756 & 40.00\% 
    & 0.774 & 0.761 & 41.25\%  
    & 0.814 & 0.789 & 50.00\%  
    & 0.843 & 0.810 & 53.75\%  \\
    \bottomrule
    \end{tabular}
    }
    \label{tab: dataset score}
\end{table*}

%% file: table/training_time_for_shadow_tokenizers.tex
\begin{table}[t]
    \centering
    \caption{Time cost (hours) for training $N$ shadow tokenizers, where each tokenizer has a vocabulary size of 200{,}000.}
    \label{training_time}
    \resizebox{0.95\columnwidth}{!}{
    \begin{tabular}{l|ccccc}
    \toprule 
    {Tokenizer Count} &$N$=1 &$N$=32 &$N$=64 &$N$=96 &$N$=128\\
    \midrule 
    {Training Time} &0.024 &0.731 &1.498 &2.251 &3.054 \\
    \bottomrule
    \end{tabular}}
\end{table}

%% file: tex/4.2_MIA_for_Dataset_with_varying_sizes.tex
\subsection{Ablation Study}

\begin{tcolorbox}[enhanced, drop shadow southwest, top=2pt,bottom=2pt]
\textbf{Finding 2.} \textit{ The membership status of the target dataset with more data samples is typically more accurately inferred by MIAs from the tokenizer.\hspace{-5cm}}
\end{tcolorbox}
\input{table/MIA_Efficiency}
\begin{figure}[t!]
    \subfigure[MIA via \vsignal]{\hspace{0.1em}\includegraphics[width=0.49\columnwidth]{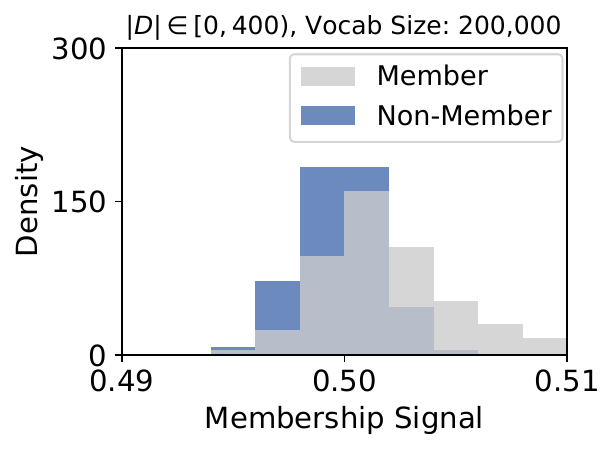}}
    \subfigure[MIA via \vsignal]{\hspace{0.1em}\includegraphics[width=0.49\columnwidth]{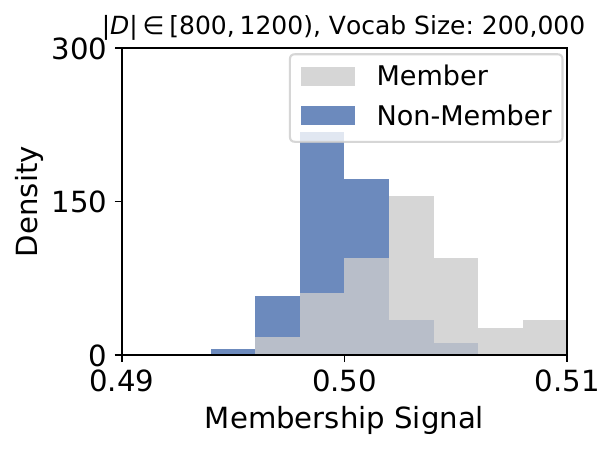}}\\
    \subfigure[MIA via \fsignal]{\hspace{0.1em}\includegraphics[width=0.49\columnwidth]{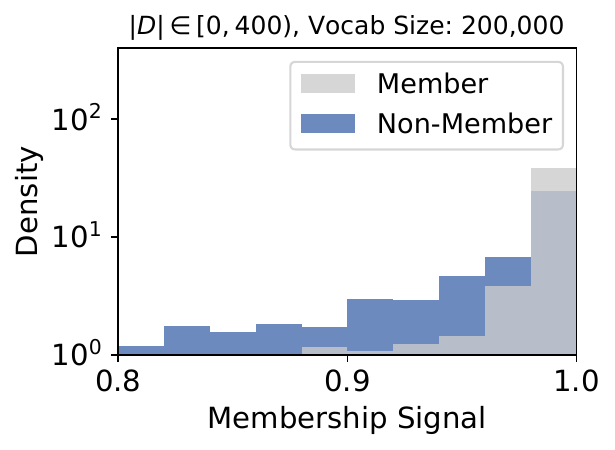}}
    \subfigure[MIA via \fsignal]{\hspace{0.1em}\includegraphics[width=0.49\columnwidth]{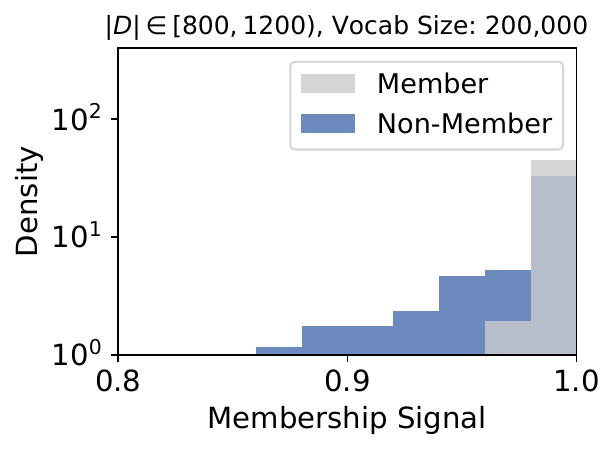}}
\caption{Distribution of \textit{members} and \textit{non-members}.}\label{fig: dataset score}
\end{figure}

\input{table/MIA_against_Defense}

\begin{figure}[t]
\hspace{-0.15cm}
\includegraphics[width=0.97\columnwidth]{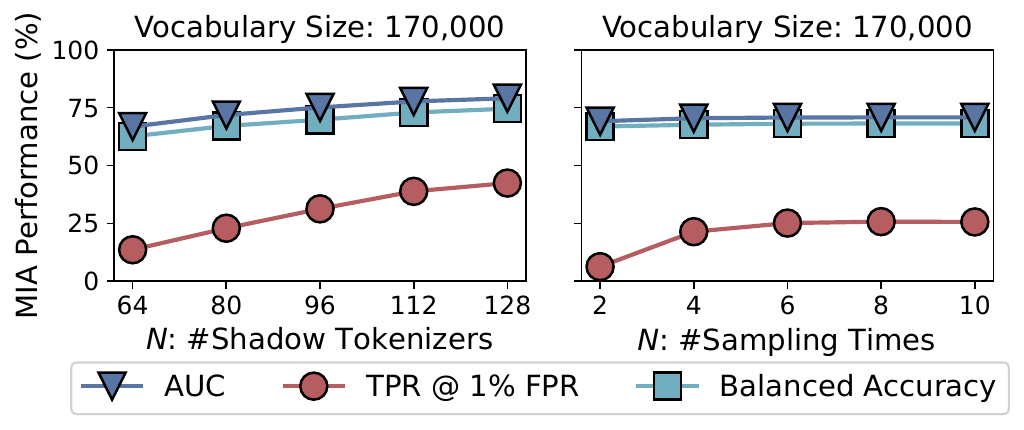}
\caption{Impact of $N$. Left: MIA via $\mathsf{Vocabulary}$ $\mathsf{Overlap}$, training $N$ shadow tokenizers. Right: MIA via $\mathsf{Frequency}$ $\mathsf{Estimation}$, sampling auxiliary datasets $N$ times.}
\label{fig: Shadow N}
\end{figure}

\mypara{Impact of Dataset Size $|D|$} Recent studies~\cite{ meeus2024did, maini2024llm, 25LLMV} have shown that increasing the amount of data used for membership inference can improve the attack performance. This finding is particularly relevant in the context of high-value litigation nowadays, where the datasets at stake are often massive~\cite{nytimes2025reddit, npr2025nyt}.
Motivated by this, we investigate whether MIAs can more effectively infer the membership of larger datasets from the tokenizers. Table~\ref{tab: dataset score} presents the performance of MIA via \vsignal and MIA via \fsignal for varying-size target datasets. The results show that both attacks become more effective as the dataset size increases, achieving particularly strong performance on large datasets. For instance, the MIA via \vsignal achieves an AUC score of 0.882 on datasets containing 800 to 1,200 data samples, whereas the MIA via \fsignal achieves a competitive AUC score of 0.843. 
Figure~\ref{fig: dataset score} further illustrates the relationship between dataset size and the membership signal in both attacks. As the dataset size increases, the overlap between the membership signal distributions for \textit{members} and \textit{non-members} decreases, thereby enhancing the discriminative power of the attacks. This reduced overlap is likely attributable to the presence of more distinctive tokens in larger datasets, which serve to strengthen the membership signal. Additional evaluations illustrating the membership signal distributions are presented in Figure~\ref{fig: full dataset score V} and Figure~\ref{fig: full dataset score F}.

\input{table/BPT_with_Defense}

\mypara{Impact of Hyperparameter $N$} We further analyze the impact of the hyperparameter $N$ on the performance of MIAs. Specifically, in the case of MIA via \vsignal, an adversary trains $N$ shadow tokenizers to capture distinctive tokens. As shown in the left plot of Figure~\ref{fig: Shadow N}, the effectiveness of this attack improves steadily as $N$ increases. However, after training more than 112 shadow tokenizers, the performance gain begins to plateau. 
In another attack, MIA via \fsignal, the adversary samples auxiliary datasets $N$ times to estimate probabilities in membership inference. As shown in the right plot of Figure~\ref{fig: Shadow N}, the effectiveness of this attack also increases with larger $N$. Nevertheless, the performance of MIA via \fsignal tends to be stable once $N$ exceeds 6. Although the results demonstrate that both attacks benefit from increasing $N$, a larger value of $N$ also incurs greater resource consumption during inference.

\input{table/Dataset_Scores_64}

%% file: table/MIA_Efficiency.tex
\begin{table}[t]
    \centering
    \caption{Time cost (hours) for MIAs via \vsignal and \fsignal inferring 4,133 target datasets.}
    \resizebox{0.95\columnwidth}{!}{
    \begin{tabular}{l|ccccc}
    \toprule
    Vocabulary Size & 80{,\,}000 & 110{,\,}000 & 140{,\,}000 & 170{,\,}000 & 200{,\,}000 \\
    \midrule
    \vsignal & 1.230 & 1.608 & 2.067 & 2.613 & 3.375 \\ 
    \fsignal & 0.170 & 0.235 & 0.303 & 0.373 & 0.432 \\
    \bottomrule
    \end{tabular}}
    \label{MIA efficiency}
\end{table}

%% file: table/MIA_against_Defense.tex
\begin{table*}[t]
    \centering
    \caption{MIAs against tokenizers with the min count defense. Here, $n_\text{min}$ denotes a threshold. If a token appears fewer than $n_\text{min}$ times in training data, it is likely a distinctive token and will be excluded from the vocabulary $\mathcal{V}_\text{target}$. $|\mathcal{V}_\text{target}|\leq 80{,}000$ indicates the vocabulary size prior to applying defense is 80{,}000. BA denotes balanced accuracy. TPR refers to TPR @ 1.0\% FPR.}

    \resizebox{1\textwidth}{!}{
    \begin{tabular}{@{}l c
                    ccc ccc ccc ccc ccc@{}}
        \toprule
        \multirow{2}{*}{\textbf{Attack Approach}}
        & \multirow{2}{*}{{\#Min Count}}
        & \multicolumn{3}{c}{{$|\mathcal{V}_\text{target}|\leq 80{,}000$}} 
        & \multicolumn{3}{c}{{$|\mathcal{V}_\text{target}|\leq 110{,}000$}} 
        & \multicolumn{3}{c}{{$|\mathcal{V}_\text{target}|\leq 140{,}000$}} 
        & \multicolumn{3}{c}{{$|\mathcal{V}_\text{target}|\leq 170{,}000$}}  
        & \multicolumn{3}{c}{{$|\mathcal{V}_\text{target}|\leq 200{,}000$}}  \\
        \cmidrule(l{7pt}r{13pt}){3-5} \cmidrule(l{6pt}r{13pt}){6-8} 
        \cmidrule(l{6pt}r{13pt}){9-11} \cmidrule(l{6pt}r{13pt}){12-14} \cmidrule(l{6pt}r{0pt}){15-17}

        &  
        & AUC & BA & TPR
        & AUC & BA & TPR 
        & AUC & BA & TPR 
        & AUC & BA & TPR 
        & AUC & BA & TPR \\
        \midrule
    \multirow{4}{*}{\vsignal}
    & w/o {defense} 
    &{0.693} & {0.666} & {26.77\% }
    & {0.718} & {0.672} & {28.75\%} 
    & {0.736 }& {0.696} & {29.72\% } 
    & {0.761} & {0.709} & {32.53\%}  
    & {0.771} & {0.711} & {34.61\%}  \\
    & w/ {$n_\text{min}=32$} 
    & 0.663 & 0.638 & 23.57\% %
    & 0.699 & 0.657 & 23.76\% %
    & 0.718 & 0.671 &  26.48\%  %
    & 0.736 & 0.686 & 28.41\%  %
    & 0.746 & 0.691 &  30.83\%  \\ %
    & w/ {$n_\text{min}=48$} 
    & 0.663 & 0.638 & 23.57\% %
    & 0.697 & 0.655 & 23.76\% %
     & 0.714 & 0.671 & 25.79\%
    & 0.717 & 0.665 &  26.09\%
    & 0.734 & 0.683 &  30.83\% \\
    & w/ {$n_\text{min}=64$} 
    & 0.663 & 0.638 &  21.93\% %
    & 0.685  & 0.640 & 23.57\%
    & 0.699  & 0.657  &  23.77\%  
    & 0.707 & 0.663 &  26.48\%
    &0.717 &  0.671 &  26.48\%  \\
    \midrule
    \multirow{4}{*}{\fsignal}
    & w/o {defense} 
    & {0.610} & {0.614} & {21.30\%} 
    & {0.645} & {0.641} & {22.00\%} 
    & {0.676} & {0.660} & {22.56\%}  
    & {0.707} & {0.681} & {25.61\% } 
    & {0.740} & {0.705} & {27.88\%}  \\
    & w/ {$n_\text{min}=32$} 
    & 0.600 &  0.602 &  18.25\% %
    & 0.633 & 0.630 & 21.30\% %
    & 0.664 & 0.647 &  22.41\%  %
    & 0.695 & 0.669 &  24.15\%  %
    & 0.730 & 0.693 &  28.80\%  \\ %
    & w/ {$n_\text{min}=48$} 
    & 0.598 & 0.598 & 18.05\% %
    & 0.633 & 0.626 & 21.06\% %
    & 0.663 & 0.645 &  21.78\% %
    & 0.690 & 0.663 &   23.23\%  
    & 0.692 & 0.664 & 25.41\%  \\
    & w/ {$n_\text{min}=64$} 
    &  0.596 & 0.600 & 17.13\% %
    & 0.630 & 0.624 & 20.72\%
    & 0.661 & 0.645 &  23.81\%  
    & 0.666 & 0.648 &  22.36\%
    & 0.668 & 0.648  &  23.72\%  \\
    \bottomrule
    \end{tabular}
    }
    \label{tab: against defense}
\end{table*}

%% file: table/BPT_with_Defense.tex
\begin{table}[t]
    \centering
\caption{Tokenizer utility measured by bytes per token. Utility scores that decrease after applying the defense are in red.}
 \resizebox{\columnwidth}{!}{
    \begin{tabular}{@{}l c
                    cccc@{}}
        \toprule
        \multirow{2}{*}{\textbf{Benchmark}}
        & \multirow{2}{*}{{$|\mathcal{V}_\text{target}|\leq$}}
        & \multicolumn{4}{c}{Tokenizer Utility Measured by Bytes per Token}  \\
        \cmidrule(l{7pt}r{1pt}){3-6} 
        
        & &  w/o {defense}  &  w/ {$n_\text{min}=32$}  &  w/ {$n_\text{min}=48$} 
        &  w/ {$n_\text{min}=64$}   \\
    \midrule
    \multirow{5}{*}{\textit{WikiText}~\cite{merity2016pointer}}
    &  $80,000$ 
     & 4.873  & 4.873~\textcolor{darkred}{(-0.0)} & 4.873~\textcolor{darkred}{(-0.0)}
    & 4.873~\textcolor{darkred}{(-0.0)}\\
    &$110,000$  
     &  4.943  &  4.943~\textcolor{darkred}{(-0.0)} &  4.943~\textcolor{darkred}{(-0.0)}
    &  4.943~\textcolor{darkred}{(-0.0)}\\
    & $140,000$ 
     &  4.986  &  4.986~\textcolor{darkred}{(-0.0)} &  4.986~\textcolor{darkred}{(-0.0)}
    &  4.986~\textcolor{darkred}{(-0.0)}\\
    & $170,000$
     &  5.008  &  5.008~\textcolor{darkred}{(-0.0)} &  5.006~\textcolor{darkred}{(-0.002)}
    &  4.995~\textcolor{darkred}{(-0.013)}\\
    & $200,000$
     &  5.025  &  5.025~\textcolor{darkred}{(-0.0)} &  5.012~\textcolor{darkred}{(-0.013)}
    &  5.000~\textcolor{darkred}{(-0.025)}\\
    \midrule
    \multirow{5}{*}{\textit{Github}~\cite{codeparrot_github_code_2025}}
    &  $80,000$ 
     & 3.740  & 3.739~\textcolor{darkred}{(-0.001)} & 3.739~\textcolor{darkred}{(-0.001)}
    & 3.738~\textcolor{darkred}{(-0.002)}\\
    &$110,000$  
    & 3.853 & 3.851~\textcolor{darkred}{(-0.002)}
    & 3.851~\textcolor{darkred}{(-0.002)} &3.849~\textcolor{darkred}{(-0.004)} \\
    & $140,000$ 
     &3.924  &3.921~\textcolor{darkred}{(-0.003)} & 3.921~\textcolor{darkred}{(-0.003)}
    &3.918~\textcolor{darkred}{(-0.006)}  \\
    & $170,000$
     &3.973 &3.970~\textcolor{darkred}{(-0.003)} & 3.967~\textcolor{darkred}{(-0.006)}
    &3.947~\textcolor{darkred}{(-0.026)}  \\
    & $200,000$
    &4.009 & 4.006~\textcolor{darkred}{(-0.003)} & 3.990~\textcolor{darkred}{(-0.019)}
    & 3.965~\textcolor{darkred}{(-0.044)}  \\
    \midrule
    \multirow{5}{*}{\textit{MGSM}~\cite{saparov2023language}}
    &  $80,000$ 
    & 3.357  & 3.356~\textcolor{darkred}{(-0.001)} & 3.356~\textcolor{darkred}{(-0.001)}
    & 3.355~\textcolor{darkred}{(-0.002)}\\
    &$110,000$  
    & 3.476 & 3.475~\textcolor{darkred}{(-0.001)}
    &3.474~\textcolor{darkred}{(-0.002)}  &3.473~\textcolor{darkred}{(-0.003)} \\
    & $140,000$ 
    & 3.601
     & 3.601~\textcolor{darkred}{(-0.0)} & 3.601~\textcolor{darkred}{(-0.0)}
    &3.600~\textcolor{darkred}{(-0.001)}  \\
    & $170,000$
     &3.663
     & 3.663~\textcolor{darkred}{(-0.0)}  & 3.662~\textcolor{darkred}{(-0.001)} 
    &3.639~\textcolor{darkred}{(-0.024)}  \\
    & $200,000$
    &3.740 
    &3.740~\textcolor{darkred}{(-0.0)} & 3.739~\textcolor{darkred}{(-0.001)}
    &3.716~\textcolor{darkred}{(-0.024)} \\
    \midrule
 \multirow{5}{*}{\textit{GPQA}~\cite{rein2024gpqa}}
    &  $80,000$ 
     & 3.925  
     & 3.924~\textcolor{darkred}{(-0.001)}  & 3.924~\textcolor{darkred}{(-0.001)}  & 3.923~\textcolor{darkred}{(-0.002)} \\
    &$110,000$  
    & 4.008 
    & 4.007~\textcolor{darkred}{(-0.001)}
    & 4.006~\textcolor{darkred}{(-0.002)} & 4.003~\textcolor{darkred}{(-0.005)} \\
    & $140,000$ 
     & 4.056
     & 4.055~\textcolor{darkred}{(-0.001)} & 4.053~\textcolor{darkred}{(-0.003)}
    & 4.050~\textcolor{darkred}{(-0.006)}  \\
    & $170,000$
     & 4.094 
     & 4.094~\textcolor{darkred}{(-0.0)} & 4.092~\textcolor{darkred}{(-0.002)}
    & 4.086~\textcolor{darkred}{(-0.008)}  \\
    & $200,000$
    &4.117 
    & 4.116~\textcolor{darkred}{(-0.001)} & 4.109~\textcolor{darkred}{(-0.008)}
    & 4.093~\textcolor{darkred}{(-0.024)}  \\
    \bottomrule
    \end{tabular}
    }
    \label{tab: BPT with defense}
\end{table}

%% file: table/Dataset_Scores_64.tex
\begin{table*}[t]
    \centering
    \caption{Impact of the target dataset size on defense mechanism ($n_\text{min}=64$). $|\mathcal{V}_\text{target}|\leq 80{,}000$ indicates the vocabulary size prior to applying defense is 80{,}000. TPR refers to TPR @ 1.0\% FPR. The results show our MIAs remain effective on large datasets.}

    \resizebox{1\textwidth}{!}{
    \begin{tabular}{@{}l c
                    ccc ccc ccc ccc ccc@{}}
        \toprule
        \multirow{2}{*}{\textbf{Attack Approach }}
        & \multirow{2}{*}{{\#Dataset Size}}
        & \multicolumn{3}{c}{\textbf{$|\mathcal{V}_\text{target}|\leq 80{,}000$}} 
        & \multicolumn{3}{c}{\textbf{$|\mathcal{V}_\text{target}|\leq 110{,}000$}} 
        & \multicolumn{3}{c}{\textbf{$|\mathcal{V}_\text{target}|\leq 140{,}000$}} 
        & \multicolumn{3}{c}{\textbf{$|\mathcal{V}_\text{target}|\leq 170{,}000$}}  
        & \multicolumn{3}{c}{\textbf{$|\mathcal{V}_\text{target}|\leq 200{,}000$}}  \\
        \cmidrule(l{7pt}r{13pt}){3-5} \cmidrule(l{6pt}r{13pt}){6-8} 
        \cmidrule(l{6pt}r{13pt}){9-11} \cmidrule(l{6pt}r{13pt}){12-14} \cmidrule(l{6pt}r{0pt}){15-17}
        & 
        & AUC & BA & TPR
        & AUC & BA & TPR 
        & AUC & BA & TPR 
        & AUC & BA & TPR 
        & AUC & BA & TPR \\
        \midrule
    \multirow{3}{*}{\vsignal}
    & {$|D|\in[0,400)$} 
    & 0.646 & 0.632 &19.16\% %
    & 0.662 & 0.636 &  21.25\%
    & 0.677 & 0.642  & 21.25\%
    & 0.683 & 0.647  &23.30\%  %
    &  0.694 & 0.655 & 24.06\%  \\%
    
    & {$|D|\in[400,800)$} 
    & 0.703 & 0.666 &  21.25\%%
    & 0.731 & 0.675 & 21.58\% 
    &  0.751 & 0.699  & 22.85\%
    &  0.761 & 0.711  & 32.50\% 
    & 0.772 & 0.717 & 37.95\%  \\

    & {$|D|\in[800,1200)$} 
    & 0.712 & 0.670 &  25.30\% %
    & 0.743 & 0.702 & 27.11\% 
    & 0.768 & 0.729  & 30.12\% 
    & 0.795 &  0.746  &  33.73\%
    & 0.797 & 0.761 & 38.75\% \\

    \midrule
    
    \multirow{3}{*}{\fsignal}
    & {$|D|\in[0,400)$} 
   & 0.591 & 0.598 & 17.50\% 
    & 0.620 & 0.619 & 18.20\%
    & 0.651 & 0.636 & 20.11\%  
    & 0.654 & 0.637 & 21.07\%
    & 0.656 & 0.639 & 21.39\% \\

    & {$|D|\in[400,800)$}
    & 0.619 & 0.615 & 22.29\% 
    & 0.672 & 0.653 & 28.31\%
    & 0.717 & 0.682 & 30.42\%  
    & 0.718 & 0.683 & 31.93\%
    & 0.723  & 0.686 & 32.83\% \\
    
    & {$|D|\in[800,1200)$} 
    & 0.734 & 0.717 & 26.25\% 
    & 0.739 & 0.731   & 33.75\% 
    & 0.744  & 0.736 &   33.75\% 
    & 0.745 & 0.742 &   35.00\% 
    &  0.748 &  0.746 & 38.75\%  \\
    \bottomrule
    \end{tabular}
    }
    \label{tab: dataset score 64}
\end{table*}

%% file: tex/4.3_MIA_against_Tokenizers_with_Min_Count_Defense.tex
\subsection{Adaptive Defense}\label{defense}

\begin{tcolorbox}[enhanced, drop shadow southwest, top=2pt,bottom=2pt]
\textbf{Finding 3.} \textit{Removing infrequent tokens from the target tokenizer’s vocabulary can partially reduce the effectiveness of MIAs. However, this mitigation comes at the cost of reduced tokenizer utility. Moreover, MIAs remain effective when inferring large datasets.}
\end{tcolorbox}
While defense against MIAs is not the primary focus of this work, we have also explored the defense mechanism to mitigate membership leakage in tokenizers.
Specifically, previous studies~\cite{25LLMV,ye2022enhanced,shokri2017membership} have demonstrated that overfitting signals are a key requirement for the success of MIAs. This suggests that methods designed to reduce overfitting may function as effective defense mechanisms~\cite{abadi2016deep,zhang2025soft}. Building on this insight, we assume an adaptive defender who reduces distinctive tokens in the tokenizer’s vocabulary.
\begin{figure}[t!]
 \hspace{-0.15cm}
\includegraphics[width=\columnwidth]{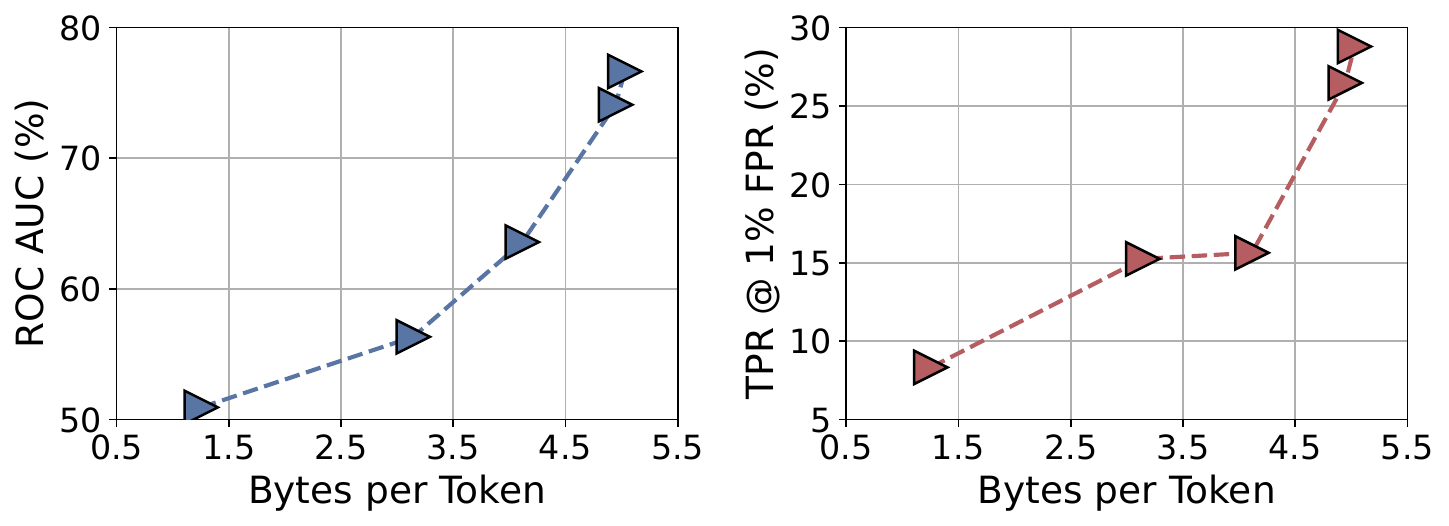}   
\caption{{Utility-privacy tradeoff, where higher attack performance (\ie AUC and TPR@1\%FPR) indicates greater privacy leakage, while higher bytes per token indicates better utility. We use \fsignal as the privacy attack.}}
\label{fig: dp}
\end{figure}

\mypara{Defender's Objective} Given the target tokenizer $f_{\mathcal{V}_\text{target}}$, the defender’s goal is to reduce the inference accuracy of MIAs on $f_{\mathcal{V}_\text{target}}$, while preserving its utility as much as possible.

\mypara{Defender's Capabilities} We assume that the defender is aware of the MIA strategy targeting the tokenizer $f_{\mathcal{V}_\text{target}}$, including the conditions that are important for their success. Specifically, our attacks rely on the distinctive tokens, which appear infrequently in the training data and overfit into the vocabulary $\mathcal{V}_\text{target}$. As a defense, the defender may modify the vocabulary $\mathcal{V}_\text{target}$ by identifying and removing such infrequent tokens. Thereby, it can mitigate the membership inference without significantly degrading the tokenizer utility.

\mypara{Min Count Mechanism} We introduce the min count mechanism as an adaptive defense against our attacks. In this mechanism, the defender post-processes the trained vocabulary $\mathcal{V}_\text{target}$ by filtering infrequent tokens. Let $n_\text{min}\in \mathbb{Z}_{> 0}$ denote the filtering threshold, and let $n_{D'}(t_i)$ represent the count of token $t_i$ in dataset $D'$. For each token $t_i\in\mathcal{V}_\text{target}$, if the aggregated count $\sum_{D' \in \mathcal{D}_\textup{mem}} n_{D'}(t_i)$ across the tokenizer’s training data is less than $n_\text{min}$, the defender removes $t_i$ from trained vocabulary $\mathcal{V}_\text{target}$. Table~\ref{tab: against defense} shows the results of MIAs against tokenizers with the min count mechanism. It is observed that, as the threshold $n_\text{min}$ increases, the defense can partially reduce the effectiveness of MIAs against tokenizers. However, this comes at the cost of the tokenizer’s utility. Table~\ref{tab: BPT with defense} shows that the compression efficiency of bytes
 per token diminishes under more strict filtering rules. While the min count mechanism can mitigate some inference risks, our MIAs remain effective, particularly for large datasets. For example, Table~\ref{tab: dataset score 64} reports an AUC of 0.797 when MIA via \vsignal infers the datasets ranging in size from 800 to 1,200 samples. 

\mypara{Differentially Private Motivated Mechanism} 
{Differential privacy (DP)~\cite{li2021large} is an effective defense mechanism against some privacy attacks.
Prior studies~\cite{abadi2016deep} have demonstrated that DP can reduce the vulnerability of membership leakage in machine learning models. 
Here we explore an approach inspired by the exponential mechanism~\cite{dong2020optimal} for DP.  Specifically, in each token-merge step, instead of choosing the pair of adjacent tokens with the highest count to merge, this approach chooses each pair with count $n$ with probability proportional to $e^{\lambda n}$ for some parameter $\lambda$.   While this mechanism is inspired by DP, we do not claim that the mechanism satisfies DP relative to the removal of one dataset, because one dataset may include multiple copies of one pair.  Instead, we focus on empirical evaluation of privacy. 
We vary the value of $\lambda$, and assess the tradeoff of utility (measured by average bytes per token) and empirical privacy risk (quantified by AUC and TPR @ 1\% FPR).  
Figure~\ref{fig: dp} shows that significant improvement in privacy (\ie lower attack performance) can be achieved only at a high utility cost (\ie lower average bytes per token).}

%% file: tex/4.4_Additional_Investigations.tex
\subsection{Additional Investigations}\label{additional_exp}
\begin{figure}[t!]
    \centering   
    \subfigure[Jaccard Similarity \label{fig: jaccard differ}]{\includegraphics[width=0.495\columnwidth]{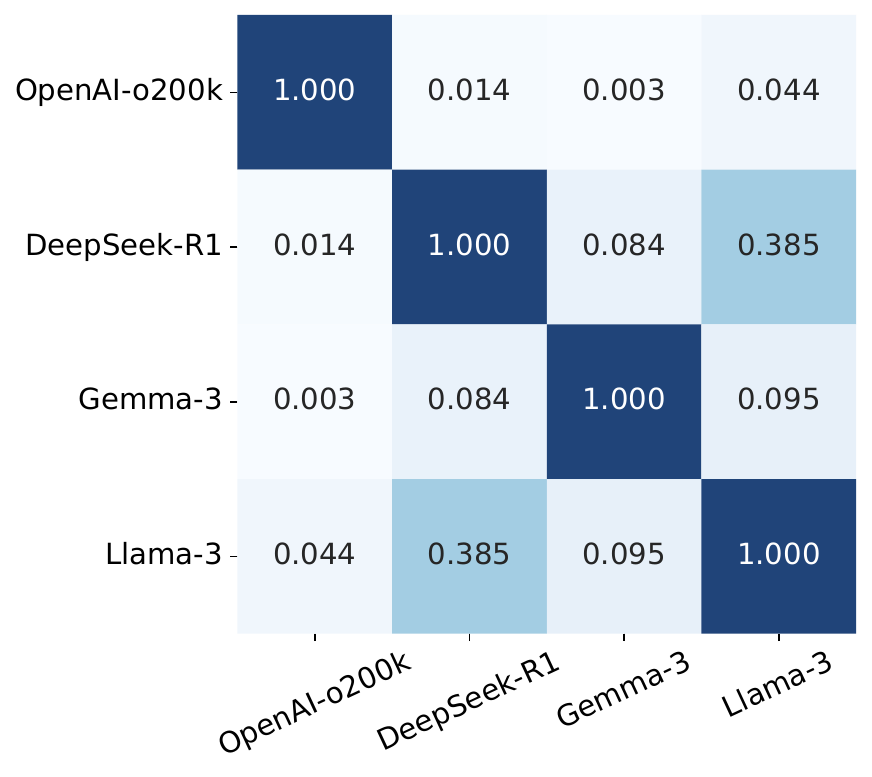}} 
    \subfigure[Differences in Merge Index \label{fig: merge differ}]{\includegraphics[width=0.495\columnwidth]{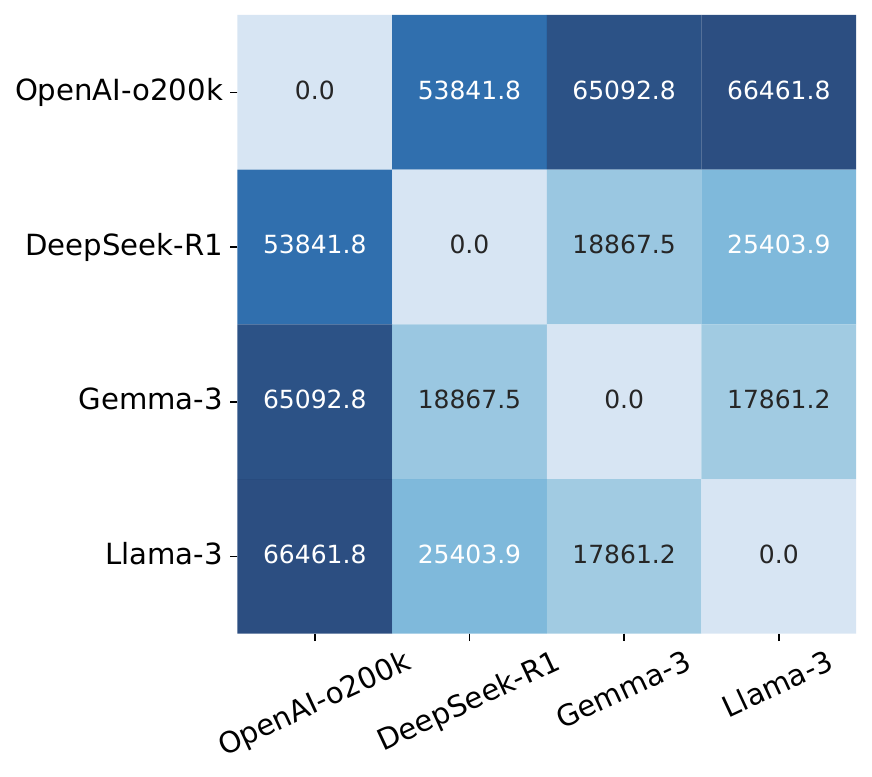}} 
\caption{Comparison of tokenizers in real-world LLMs.}
    \label{fig: exploration}
\end{figure}

\mypara{Distinctive Tokens in Real-world Tokenizers} Our MIAs exploit distinctive tokens, which can vary across different tokenizers, considering both their occurrence and their merge indices. A key question is whether real-world tokenizers differ in such ways that might enable effective MIAs. To investigate this, we compare the vocabularies of tokenizers used in real-world LLMs~\cite{openai2025tiktoken,karter2024gemini,guo2025deepseek,touvron2023llama}. Specifically, we limit our analysis to the first 120,000 tokens of each tokenizer to enable a fair comparison across varying sizes. Figure~\ref{fig: jaccard differ} presents the similarity between tokenizers' vocabularies based on the Jaccard index~\cite{bag2019efficient}. If there were no distinctive tokens, we would expect high similarity scores between any pair of tokenizers. However, the evaluation results show that real-world tokenizers contain a significant number of distinctive tokens, as evidenced by the maximum Jaccard index of only 0.385 observed between DeepSeek-R1~\cite{guo2025deepseek} and Llama-3~\cite{touvron2023llama}. In addition, Figure~\ref{fig: merge differ} illustrates the average absolute differences in the merge index for tokens between any two vocabularies. If there were no significant differences between the two vocabularies, we would expect only slight variations in the token merge indices. Nevertheless, the results suggest that token merge indices also vary largely in real-world tokenizers.

We further evaluate MIAs on out-of-distribution datasets, a single document containing unique tokens (canaries), and non-BPE tokenizers. Please refer to Appendix~\ref{add_eval} for details.


%% file: tex/5._Related_Work.tex
\section{Related Work}\label{related_work}
\mypara{Tokenizers} Tokenizers play a crucial role in enabling the generation and comprehension capabilities of LLMs~\cite{guo2025deepseek}. By converting raw text into discrete tokens, tokenizers provide the input representations that models require for inference~\cite{zouhar2023formal}. Recent research has highlighted the connection between tokenizers and scaling laws, suggesting that larger models benefit from a larger tokenizer vocabulary, leading to improved performance under the same training cost~\cite{huangover, tao2024scaling}. Nevertheless, our work reveals that scaling up the tokenizer's vocabulary increases its vulnerability to MIAs. We underscore the necessity of paying attention to these overlooked risks.

\mypara{Membership Inference} Membership inference attacks~\cite{zarifzadeh2024low} are designed to determine whether a specific entity was included in the training data of a machine learning (ML) model. These attacks have become fundamental tools for quantifying privacy leakage in various scenarios~\cite{zarifzadeh2024low,carlini2021extracting}.
The first  MIA is proposed by Shokri \et~\cite{shokri2017membership}, which focuses on demining record-level membership from ML-based classifiers. Recent studies underscore the importance of detecting whether specific data was used during the pre-training of LLMs, as such data may contain copyrighted or sensitive information. To address these risks, several membership inference methods have been proposed.
Shi \et\cite{shi2023detecting} introduced MIN-K\% PROB, a technique that detects membership by identifying outlier words with unusually low probabilities in previously unseen examples. 
Duarte \et\cite{duarte2024cop} proposed DE-COP, which probes LLMs using multiple-choice questions that include both verbatim and paraphrased versions of candidate sentences. 
However, these MIA methods face significant evaluation challenges. 
Duan \et \cite{duanmembership} argue that existing benchmarks suffer from a temporal distribution shift between \textit{members} and \textit{non-members}, potentially invalidating evaluation results. 
Meeus \et \cite{meeus2024sok} further observe that many methods may incorporate mislabeled samples and use impractical evaluation models in their experiments.



%% file: tex/6._Discussion.tex
\section{Discussion}\label{discussion}

\mypara{LLM Dataset Inference}
Recent research~\cite{maini2021dataset} has emphasized the importance of dataset inference in LLMs. In particular, existing methods~\cite{maini2024llm,choicontaminated} attempt to predict whether a specific dataset was used to train a target LLM by analyzing the model's output. However, as noted in Section~\ref{sec: intro}, these methods face significant challenges during evaluation, such as mislabeled samples, distribution shifts, and mismatches in target model sizes compared to real-world models. Furthermore, these attacks typically introduce additional assumptions about the adversary, such as access to model output loss or fine-tuning the target model, which are not guaranteed to hold in closed-source LLMs. Moreover, they can be defended by adding noise during the model training process~\cite{zhang2025soft}.

\mypara{Limitations of Tokenizer Inference}
While the tokenizers listed in Table~\ref{Tokenizer Justification} are trained using the models' pretraining data, some tokenizers might not follow this paradigm. For instance, ChatGLM~\cite{glm2024chatglm} employs OpenAI's tokenizer. In such cases, MIAs on these tokenizers will only reveal the tokenizer's training data rather than LLMs' pretraining corpus. Moreover, due to the absence of ground-truth training data for commercial tokenizers, we are unable to evaluate our attacks on them. Instead, we conduct evaluations on our trained tokenizers with vocabulary sizes and utility comparable to those of commercial tokenizers (see Figure~\ref{fig: cmp tokenizers}). Notably, this limitation is not unique to our work but is a common challenge in the field of membership inference, as also observed in other related studies~\cite{liu2024please,li2025enhanced,hui2021practical,maini2024llm}. Furthermore, our attack evaluations are restricted to tokenizers used in LLMs. As a result, the feasibility of MIAs on tokenizers for classification models~\cite{liu2019roberta} and diffusion-based language models~\cite{sahoo2024simple} remains unexplored.

%% file: tex/7._Conclusion.tex
\section{Conclusion}\label{conclusion}
In this paper, we review the limitations of existing MIAs against pre-trained LLMs and introduce the tokenizer as a new attack vector to address these challenges. To demonstrate its feasibility for membership inference, we present the first study of MIAs on tokenizers of LLMs. By analyzing overfitting signals during tokenizer training, we proposed five attack methods for inferring dataset membership. 
Extensive evaluations on millions of Internet data demonstrate that our shadow-based attacks achieve strong performance. 
To mitigate these attacks, we further propose an adaptive defense mechanism. 
Although our proposed defense can reduce the membership leakage, it does so at the cost of tokenizer utility. 
Our findings highlight the vulnerabilities associated with LLMs' tokenizers. Through this endeavor, we hope our research contributes to the design of privacy-preserving tokenizers, towards building secure machine learning systems.

%

%% file: tex/Ethic_Consideration.tex
\section*{Ethical Considerations}
In this work, we conduct the first study of membership leakage through tokenizers. We recognize that it is our responsibility to carefully assess the ethical implications of our research. Our analysis follows the stakeholder-based framework and is guided by the principles outlined in The Menlo Report.

\mypara{Stakeholder-based Analysis} We identify several key stakeholders, which might be impacted by this research:

\begin{itemize}
\item \mysubpara{LLM Developers} Our work focuses on the tokenizers of LLMs and closely involves their developers. Leveraging our proposed methods, developers can estimate the privacy risk in the tokenizer before making it publicly accessible. This is essential as tokenizers in commercial LLMs are typically open-sourced to support transparent billing.
\item \mysubpara{Data Subjects} The introduction of a new attack vector for MIAs could be used to infer personal information and lead to privacy violations. However, our experiments solely utilize public,non-sensitive datasets collected by Google. We emphasize that our proposed attacks are never intended to facilitate or encourage the privacy leakage from LLMs. 
\item \mysubpara{Adversaries} This work could potentially be used by adversaries to exploit privacy vulnerabilities. In recognition of this, we propose an adaptive defense to mitigate such risk.
\end{itemize}

\mypara{Mitigation} To mitigate the membership leakage in tokenizers, we introduce the min count mechanism as an adaptive defense. Specifically, our attacks exploit distinctive tokens that occur rarely in the training data and are overfit into the tokenizer's vocabulary. The min count mechanism counters this by removing these infrequent tokens from the vocabulary. {Furthermore, we explore a differentially private motivated mechanism at each token merging step during tokenizer training. The experimental results demonstrate that both of our approaches reduce privacy risks.}  

\mypara{Justification for Research and Publication} 
The vulnerabilities exploited by our MIAs stem from the current practice of open-sourcing LLM tokenizers. Since similar attack techniques could be independently discovered by adversaries, it is important to disclose and analyze these issues within the research community. Moreover, by publishing our findings, LLM developers can address such vulnerabilities proactively.

We emphasize that our proposed attacks are not intended to facilitate or encourage privacy leakage from LLMs. We hope our work contributes to the design of privacy-preserving tokenizers, laying the groundwork for secure and trustworthy systems that benefit society.


%% file: tex/Open_Science.tex
\section*{Open Science}

All attack and defense codes, as well as datasets, are available at
\url{https://github.com/mengtong0110/Tokenizer-MIA} and \url{https://zenodo.org/records/20338916}.  

\section*{Acknowledgments}
We thank all the anonymous reviewers and shepherd for their valuable comments. The work of Meng Tong, Kejiang Chen, and Weiming Zhang was supported by National Natural Science Foundation of China under Grant U2336206 and 62472398.

%% file: tex/appendix.tex
\input{table/Dataset_Scores_32}

\input{table/Dataset_Scores_48}
\begin{table}[t]
    \centering
    \caption{Additional evaluations with $|\mathcal{V}_\text{target}|=200{,}000$.}
    \label{a_exp}
    \resizebox{\columnwidth}{!}{
    \begin{threeparttable}  

    \vspace{-0.1cm}
    
    \begin{tabular}{lccc}
    \toprule 
    \textbf{Evaluation Settings} & \textbf{AUC} & \textbf{TPR @ 1\% FPR} & \textbf{ Balanced Acc.} \\
    \midrule 
    \textit{Canary Experiments}  \\
    \hspace{0.15cm}+\hspace{0.15cm}MIA via \vsignal\hspace{0.15cm}  &1.000 &100.00\%  &1.000\\
    \hspace{0.15cm}+\hspace{0.15cm}MIA via \fsignal\hspace{0.15cm}  &1.000 &100.00\%  &1.000\\
    \midrule 
    \textit{Out-of-distribution Datasets}  \\
    \hspace{0.15cm}+\hspace{0.15cm}MIA via \vsignal\hspace{0.15cm}  &0.757 &21.00\%  &0.730\\
    \hspace{0.15cm}+\hspace{0.15cm}MIA via \fsignal\hspace{0.15cm}  &0.739 &13.00\%  &0.715\\
    \midrule 
    \textit{Unigram Tokenizers}  \\
    \hspace{0.15cm}+\hspace{0.15cm}MIA via \vsignal\hspace{0.15cm}  &0.762 &23.00\%  &0.712\\
    \hspace{0.15cm}+\hspace{0.15cm}MIA via \fsignal\hspace{0.15cm}  &0.764 &30.00\%  &0.728\\
    \bottomrule
    \end{tabular}
    \end{threeparttable} 
        }
        \vspace{-0.1cm}
\end{table}

\input{table/Dataset_Scores_64_a}

\section{Proof of Theorem 4.2}\label{a: proof}


\begin{proof}
Given that $\Pr(t_i \mid \mathcal{V}_\text{target}) \propto \frac{1}{i^\alpha}$, let us assume 
\begin{equation}\label{op1}
\Pr(t_i \mid \mathcal{V}_\text{target}) = C/i^\alpha, 
\end{equation}
where $C \in \mathbb{R}_{> 0}$ is a constant. Since the sum of the frequencies $\Pr(t_i \mid \mathcal{V}_\text{target})$ for all $t_i\in \mathcal{V}_\text{target}$ is at most 1, it follows that
\begin{equation}
1 \geq \hspace{-0.3em}\sum_{j = x_\textup{min} + 1}^{|\mathcal{V}_\text{target}|} \Pr(t_j \mid \mathcal{V}_\text{target}) = \hspace{-0.3em}\sum_{j=1+x_\textup{min}}^{|\mathcal{V}_\text{target}|}\hspace{-0.7em}{C}/{j^\alpha}.
\end{equation}
Therefore, the upper bound for the constant $C$ is
\begin{equation}\label{up1}
C \leq \frac{1}{\sum_{j = x_\textup{min} + 1}^{|\mathcal{V}_\text{target}|} \frac{1}{j^\alpha}}.
\end{equation}
Using Equations~\ref{op1}~and~\ref{up1}, we derive the upper bound for the frequency $\Pr(t_i \mid \mathcal{V}_\text{target})$:
\begin{equation}
\Pr(t_i \mid \mathcal{V}_\text{target}) \leq \frac{1}{\sum_{j = x_\textup{min} + 1}^{|\mathcal{V}_\text{target}|} \frac{1}{j^\alpha}} \cdot \frac{1}{i^\alpha}.
\end{equation}
Taking the negative logarithm, the self-information of $t_i$ is bounded from below:
\begin{equation}
\operatorname{SI}(t_i, \mathcal{V}_\text{target}) = -\log \Pr(t_i \mid \mathcal{V}_\text{target}) \geq \log(\hspace{-0.9em}\sum_{j=x_\textup{min}+1}^{|\mathcal{V}_\text{target}|}\hspace{-0.9em}{i^\alpha}/{j^\alpha}).
\end{equation} 
Therefore, the RTF-SI satisfies:
\begin{equation}
\operatorname{RTF\text{-}SI}( D, t_i, \mathcal{V}_\text{target}) \geq \frac{n_D(t_i)}{\sum_{D' \in \mathbb{D}} n_{D'}(t_i)} \cdot \log (\hspace{-0.9em}\sum_{j=x_\textup{min}+1}^{|\mathcal{V}_\text{target}|}\hspace{-0.9em}{i^\alpha}/{j^\alpha}).
\end{equation}
This theorem allows an adversary to approximate the RTF-SI using its lower bound under a power-law distribution.
\end{proof}

\section{Additional Evaluations}\label{add_eval}

\mypara{Canary Experiments} Our MIAs focus on dataset-level membership. A key question is whether their effectiveness implies document-level risk~\cite{carlini2021extracting}. To this end, we conduct canary evaluations~\cite{zhang2025position,aerni2024evaluations}. Specifically, we synthesize a token from five alphabets and insert it into a single document 40 times. We use MIAs that perform well on datasets to infer the membership of such documents. Table~\ref{a_exp} shows that the canary cases are highly detectable. This can be attributed to our attacks, which identify unique tokens, thus establishing a one-to-one mapping between the tokens and the documents.

\mypara{Out-of-distribution Datasets} We further evaluate our MIAs using auxiliary datasets from corpora different from those used in the tokenizer’s training data. Specifically, we train the tokenizer on the \textit{PILE} corpus~\cite{biderman2022datasheet} (published in 2022) with the target dataset to infer membership. The auxiliary data is drawn from the \textit{C4} corpus~\cite{raffel2020exploring} (published in 2020), ensuring that no samples are shared between the auxiliary data and the tokenizer’s training data. The results in Table~\ref{a_exp} demonstrate that our MIAs perform effectively without relying on the shared data assumption in the threat model. Indeed, previous work~\cite{carlini2022membership} has shown that shadow-based MIAs with out-of-distribution auxiliary data can successfully infer membership.

\mypara{Unigram Tokenizers} Our MIAs focus on tokenizers trained using the BPE algorithm. This is because state-of-the-art LLMs employ this training paradigm~\cite{dai2024deepseekmoe}. In addition to BPE tokenizers, we further evaluate the MIAs against unigram tokenizers\footnote{\url{https://github.com/google/sentencepiece}}. Table~\ref{a_exp} presents the experimental results. It is observed that our MIAs also generalize to unigram tokenizers. This is likely because the unigram tokenizer merges the most frequent string as a new token at each step, which resembles the BPE and overfits distinctive tokens into the vocabulary.

 \FloatBarrier

\begin{figure*}[t]
    \centering
    \subfigure[Vocabulary Size: $80\text{,\,}000$]{\includegraphics[width=0.196\textwidth]{fig/intial_80000_obs.pdf}}
    \subfigure[Vocabulary Size: $110\text{,\,}000$]{\includegraphics[width=0.196\textwidth]{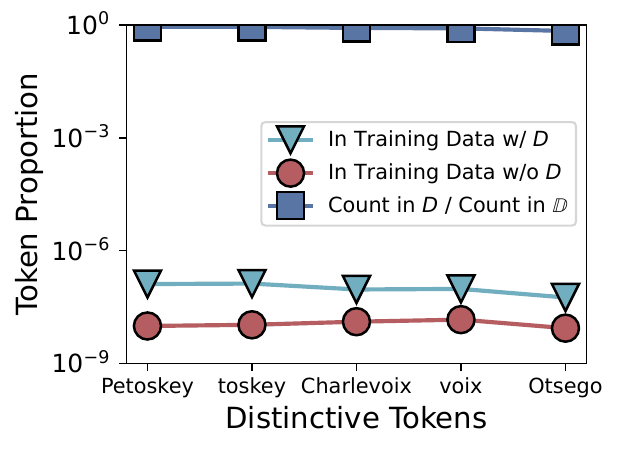}}
    \subfigure[Vocabulary Size: $140\text{,\,}000$]{\includegraphics[width=0.196\textwidth]{fig/intial_140000_obs.pdf}}
    \subfigure[Vocabulary Size: $170\text{,\,}000$]{\includegraphics[width=0.196\textwidth]{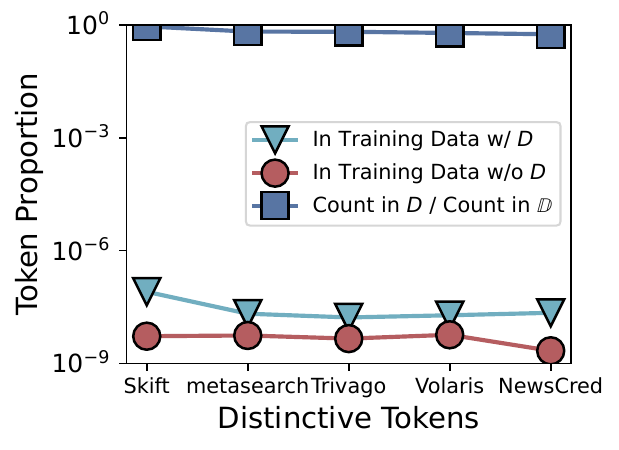}}
    \subfigure[Vocabulary Size: $200\text{,\,}000$]{\includegraphics[width=0.196\textwidth]{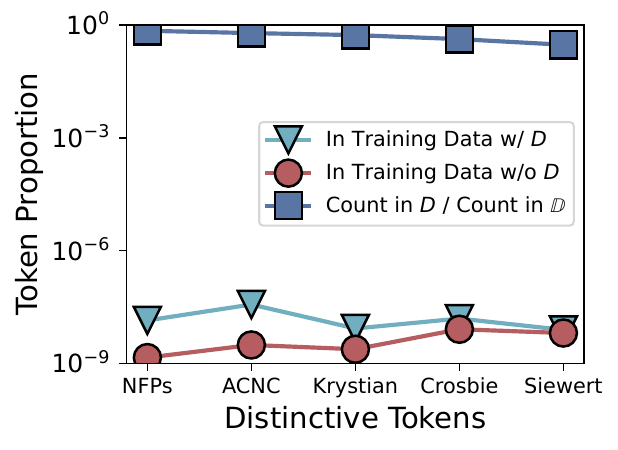}}
    \vspace{-0.1cm}
\caption{Distinctive tokens in MIA via \vsignal.}
    \label{fig: full diff}
\end{figure*}
\begin{figure*}[h]
    \centering
    \subfigure[Vocabulary Size: $80\text{,\,}000$]{\includegraphics[width=0.196\textwidth]{fig/roc_curve_vsize_80000.pdf}}
    \subfigure[Vocabulary Size: $110\text{,\,}000$]{\includegraphics[width=0.196\textwidth]{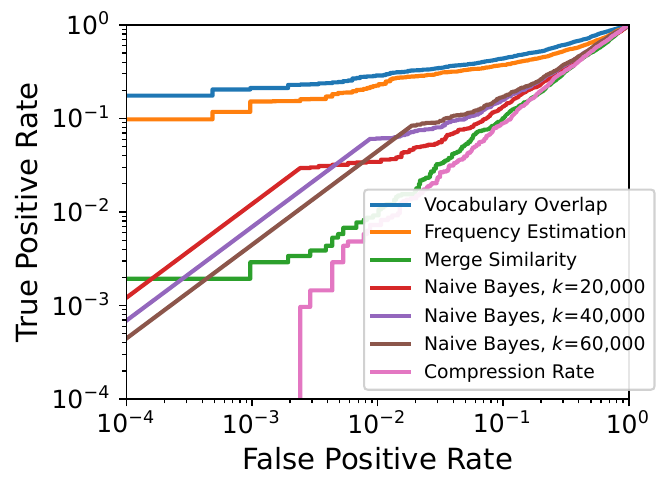}}
    \subfigure[Vocabulary Size: $140\text{,\,}000$]{\includegraphics[width=0.196\textwidth]{fig/roc_curve_vsize_140000.pdf}}
    \subfigure[Vocabulary Size: $170\text{,\,}000$]{\includegraphics[width=0.196\textwidth]{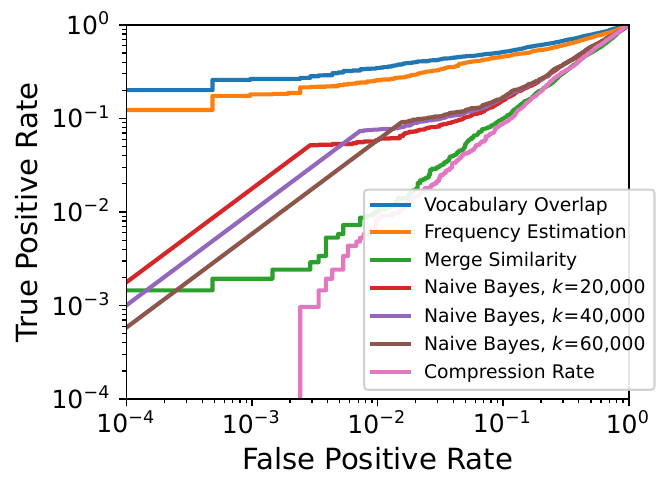}}
    \subfigure[Vocabulary Size: $200\text{,\,}000$]{\includegraphics[width=0.196\textwidth]{fig/roc_curve_vsize_200000.pdf}}
    \vspace{-0.1cm}
\caption{ROC curves for MIAs using different methods.}
    \label{fig: full ROC curves}
\end{figure*}

\begin{figure*}[h]
    \centering
    \subfigure[Vocabulary Size: $80\text{,\,}000$]{\includegraphics[width=0.196\textwidth]{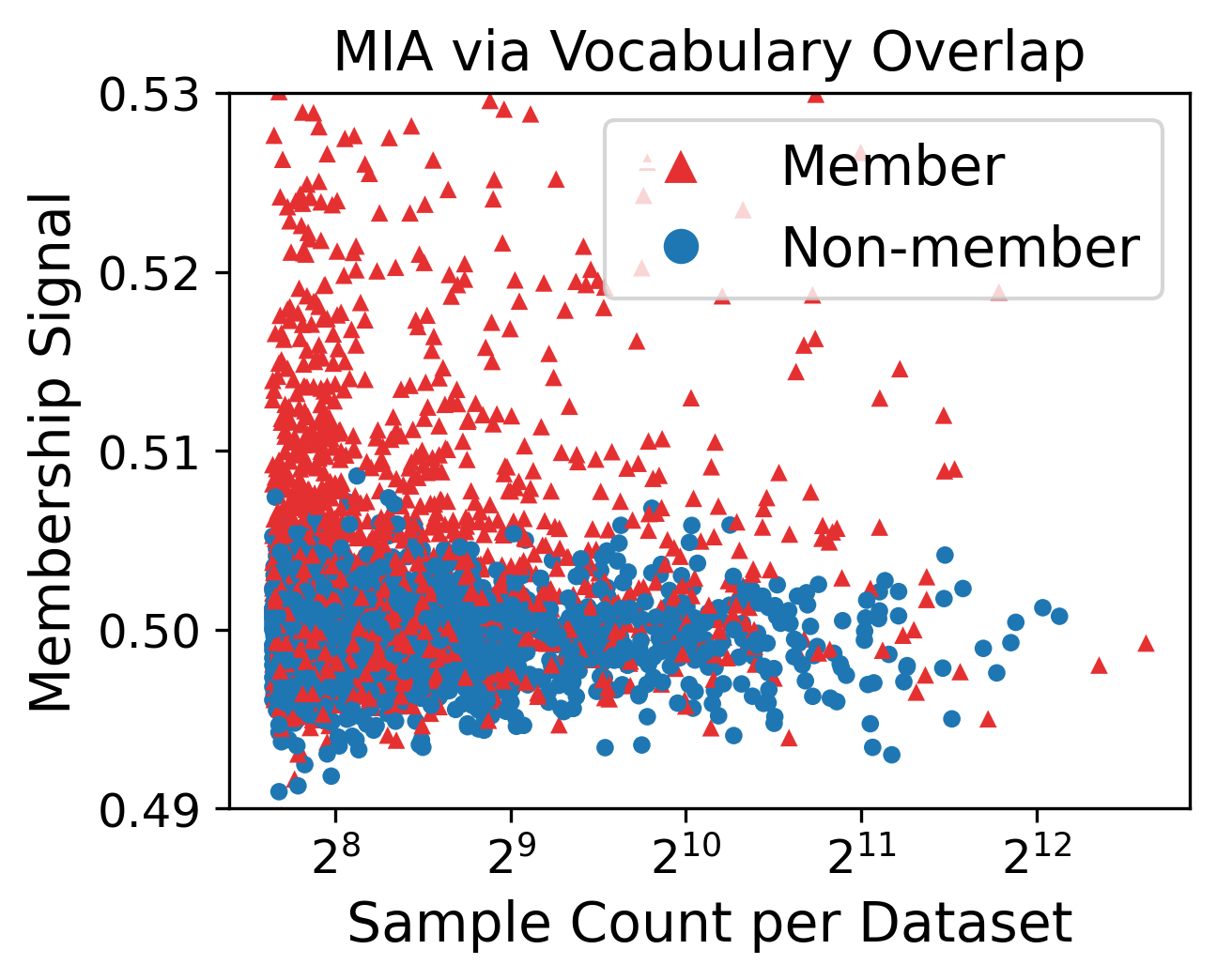}}
    \subfigure[Vocabulary Size: $110\text{,\,}000$]{\includegraphics[width=0.196\textwidth]{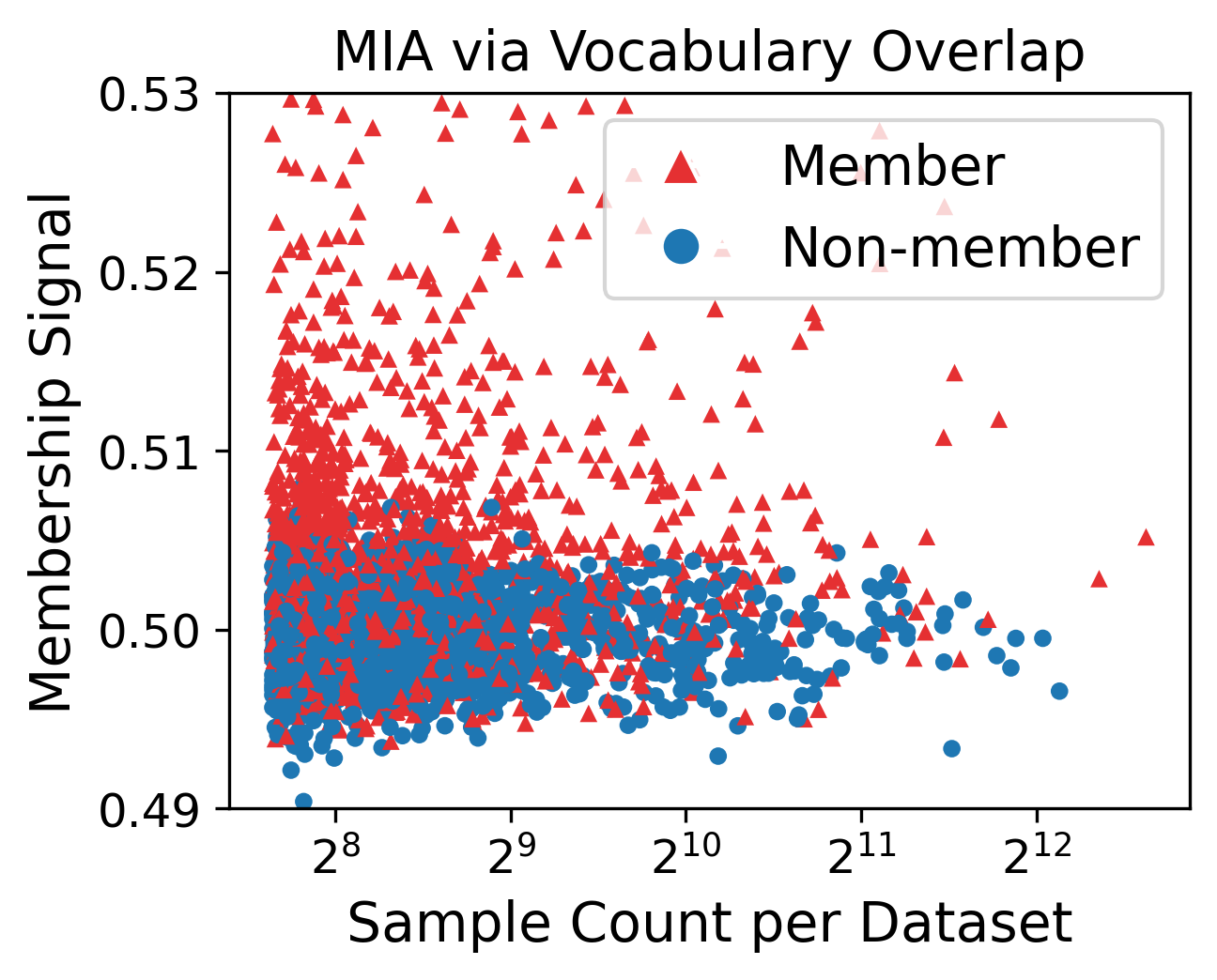}}
    \subfigure[Vocabulary Size: $140\text{,\,}000$]{\includegraphics[width=0.196\textwidth]{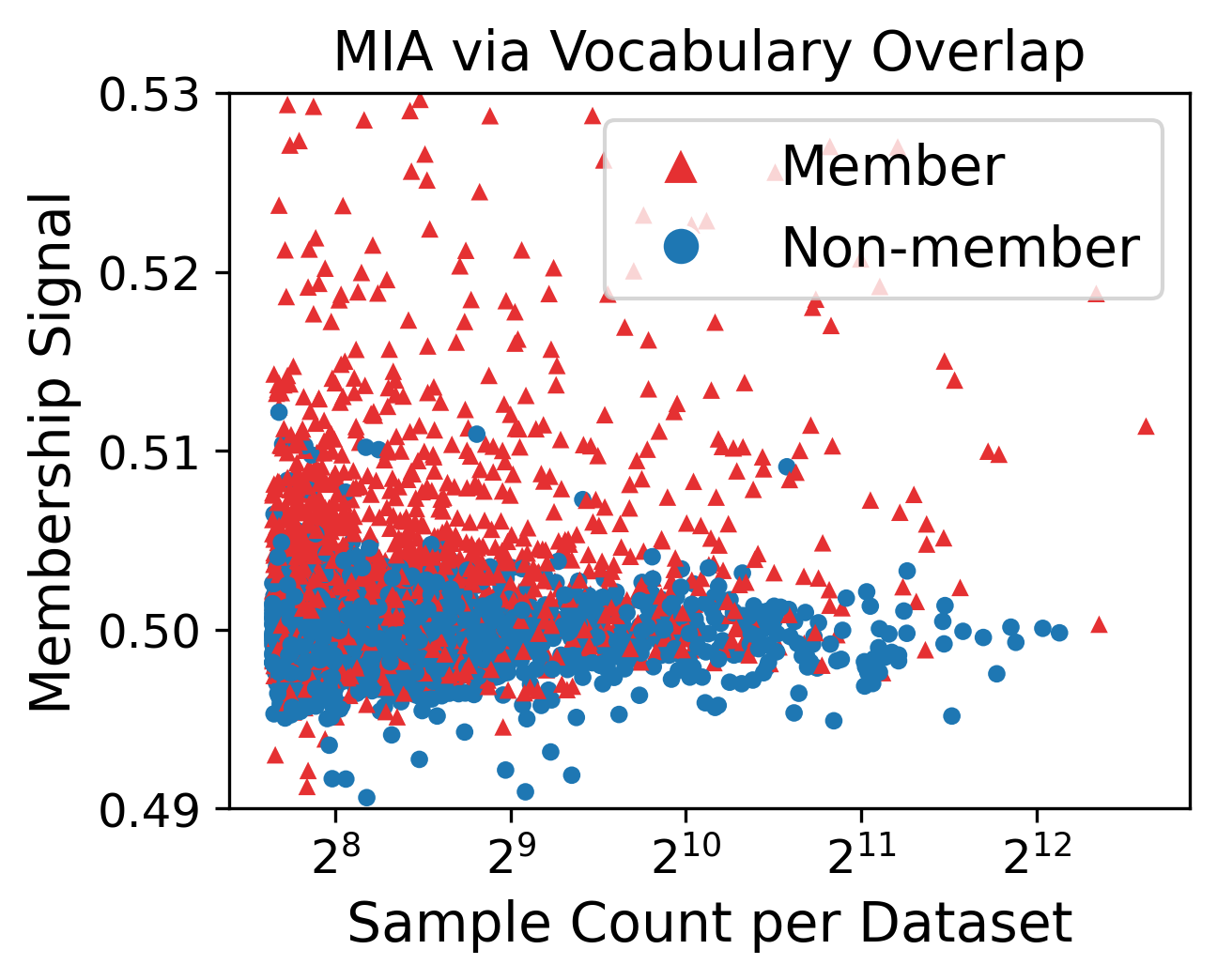}}
    \subfigure[Vocabulary Size: $170\text{,\,}000$]{\includegraphics[width=0.196\textwidth]{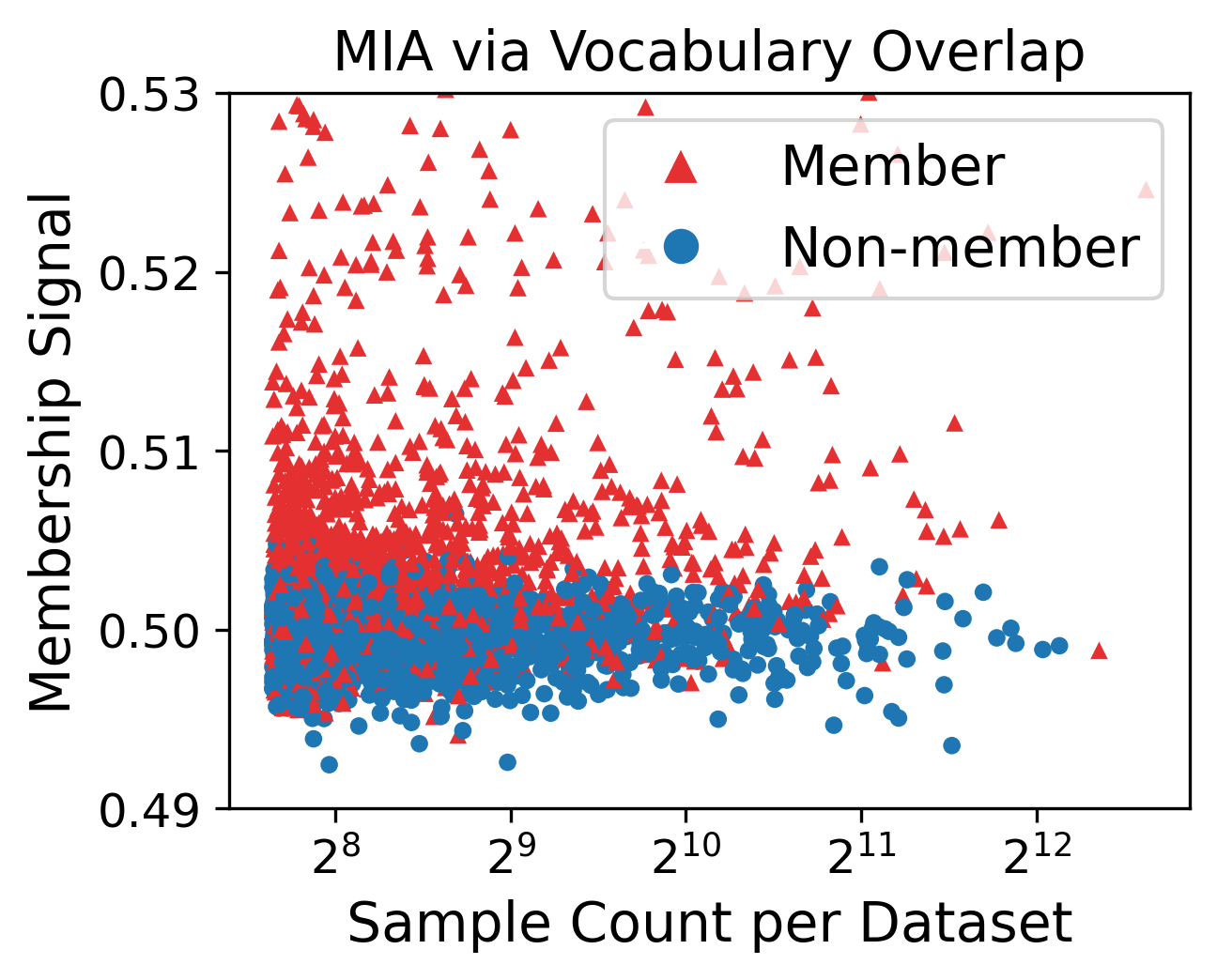}}
    \subfigure[Vocabulary Size: $200\text{,\,}000$]{\includegraphics[width=0.196\textwidth]{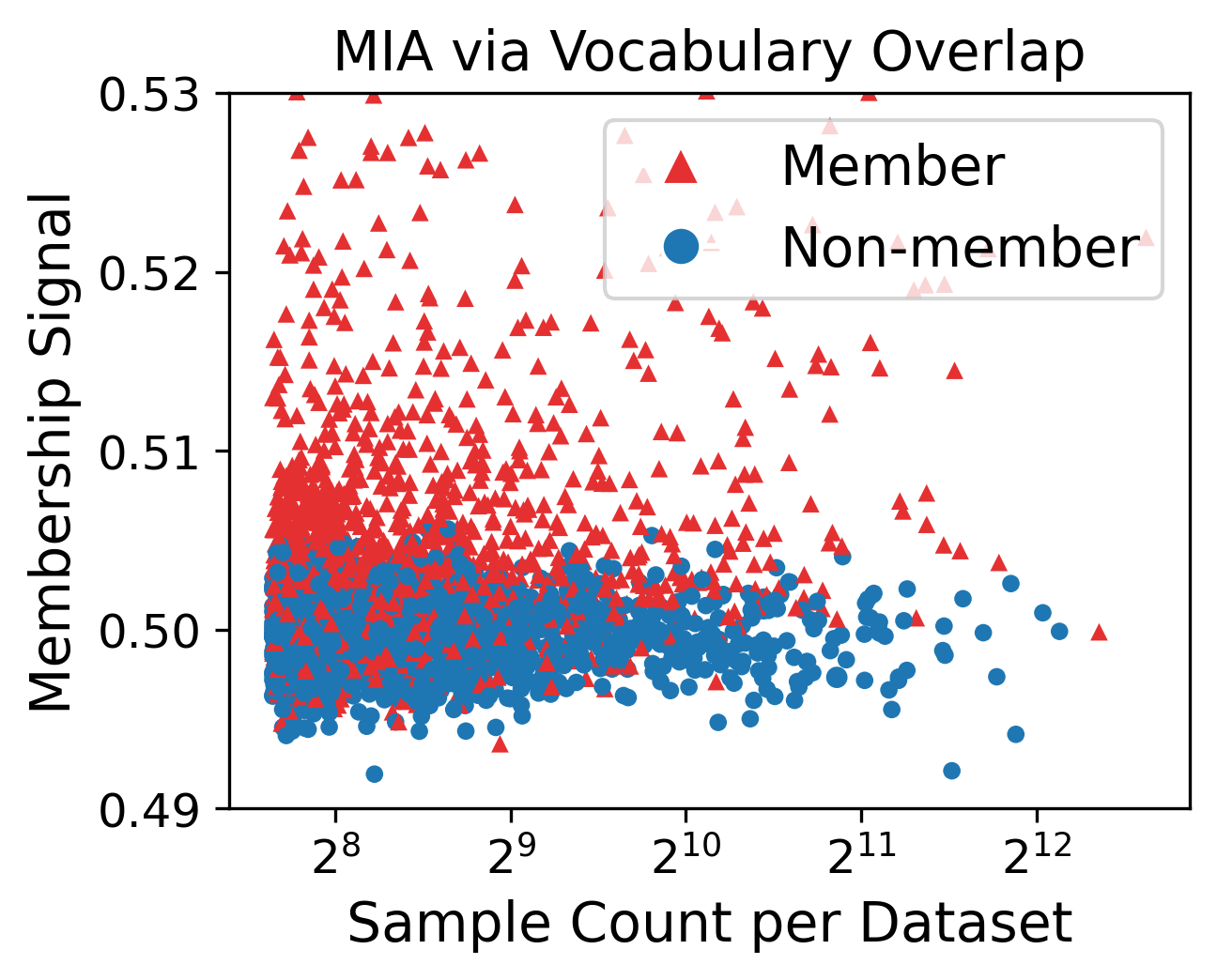}}
    \vspace{-0.1cm}
\caption{Dataset distribution based on MIA via \vsignal.}
    \label{fig: full dataset score V}
\end{figure*}

\begin{figure*}[h]
    \centering
    \subfigure[Vocabulary Size: $80\text{,\,}000$]{\includegraphics[width=0.196\textwidth]{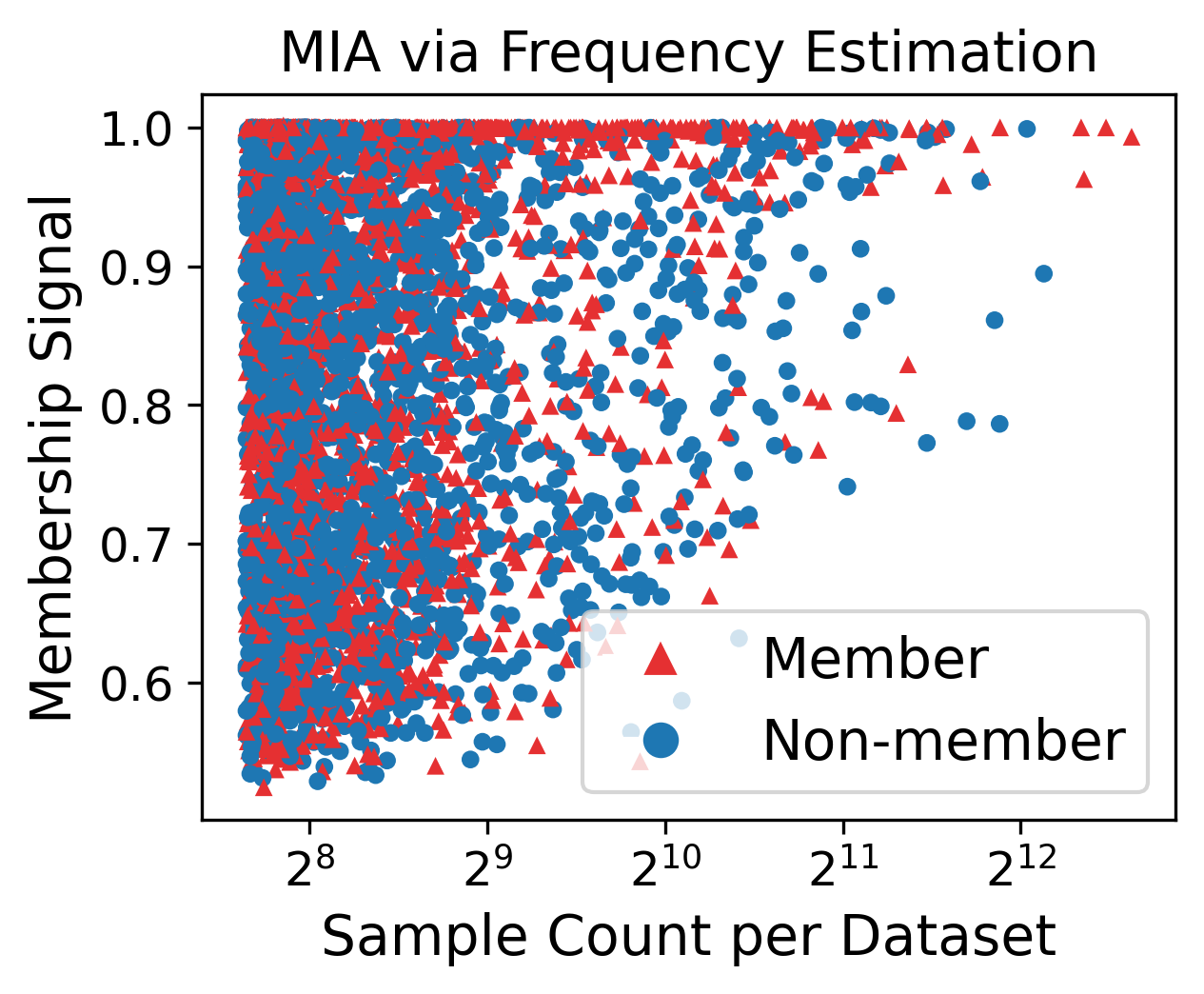}}
    \subfigure[Vocabulary Size: $110\text{,\,}000$]{\includegraphics[width=0.196\textwidth]{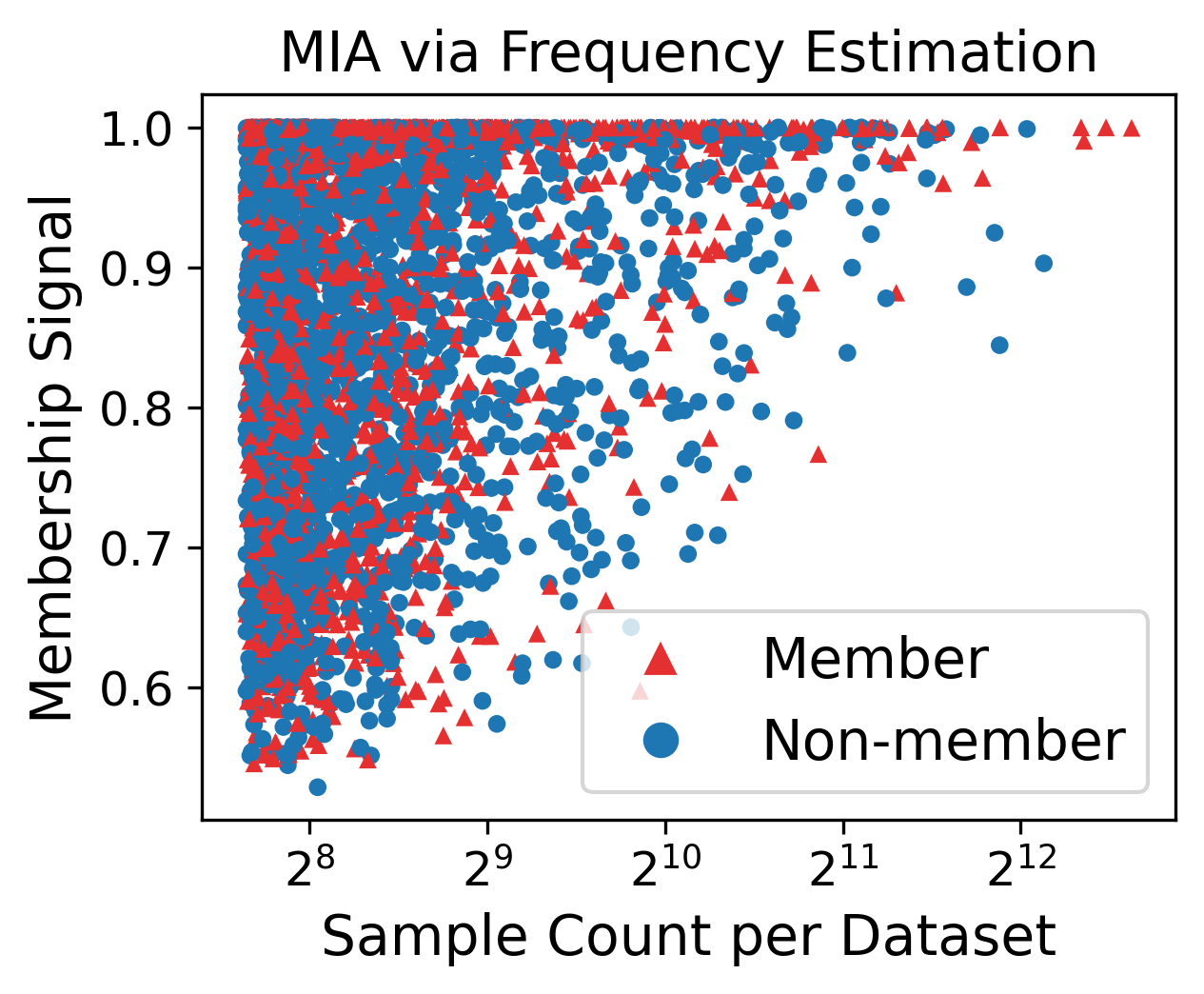}}
    \subfigure[Vocabulary Size: $140\text{,\,}000$]{\includegraphics[width=0.196\textwidth]{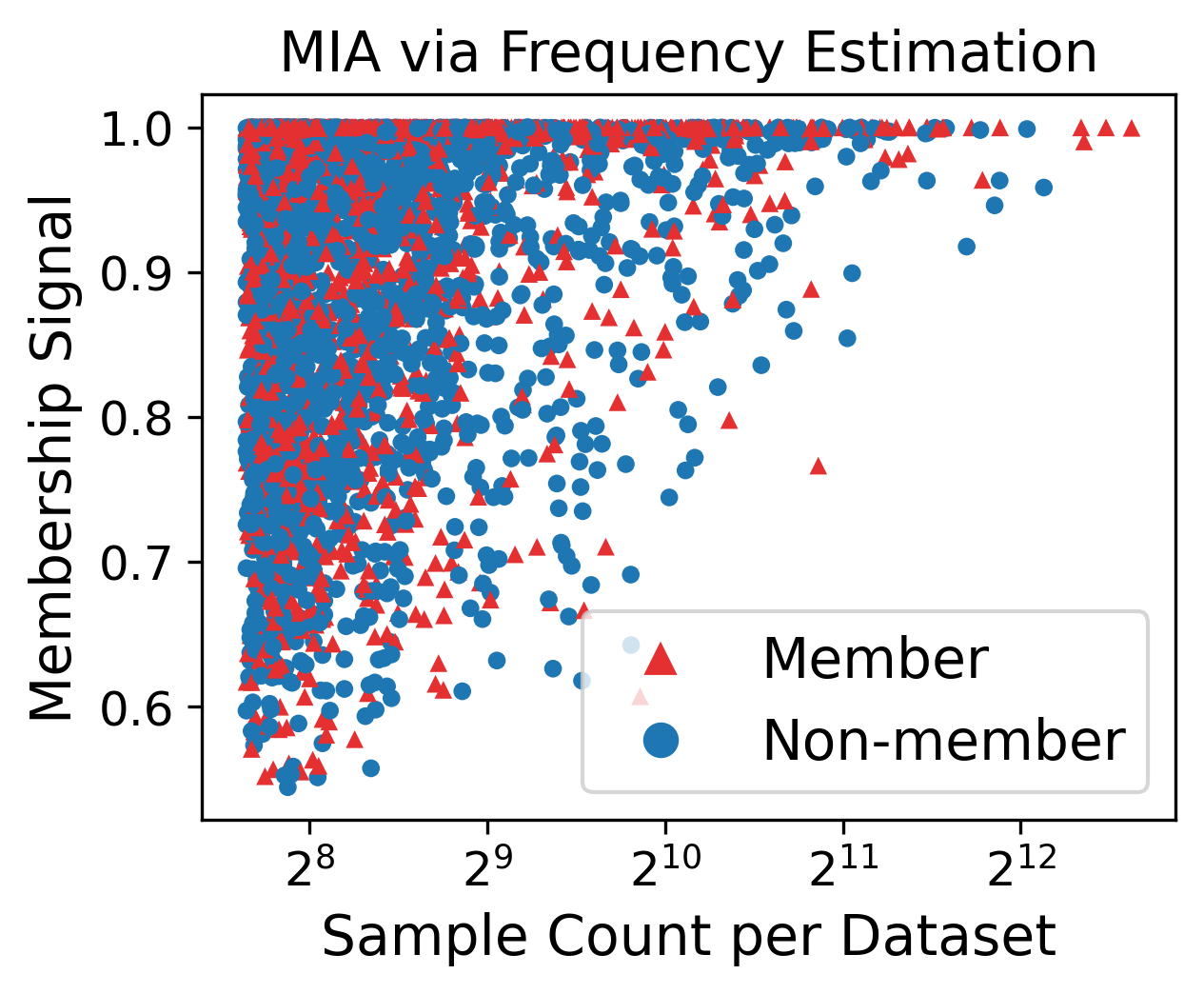}}
    \subfigure[Vocabulary Size: $170\text{,\,}000$]{\includegraphics[width=0.196\textwidth]{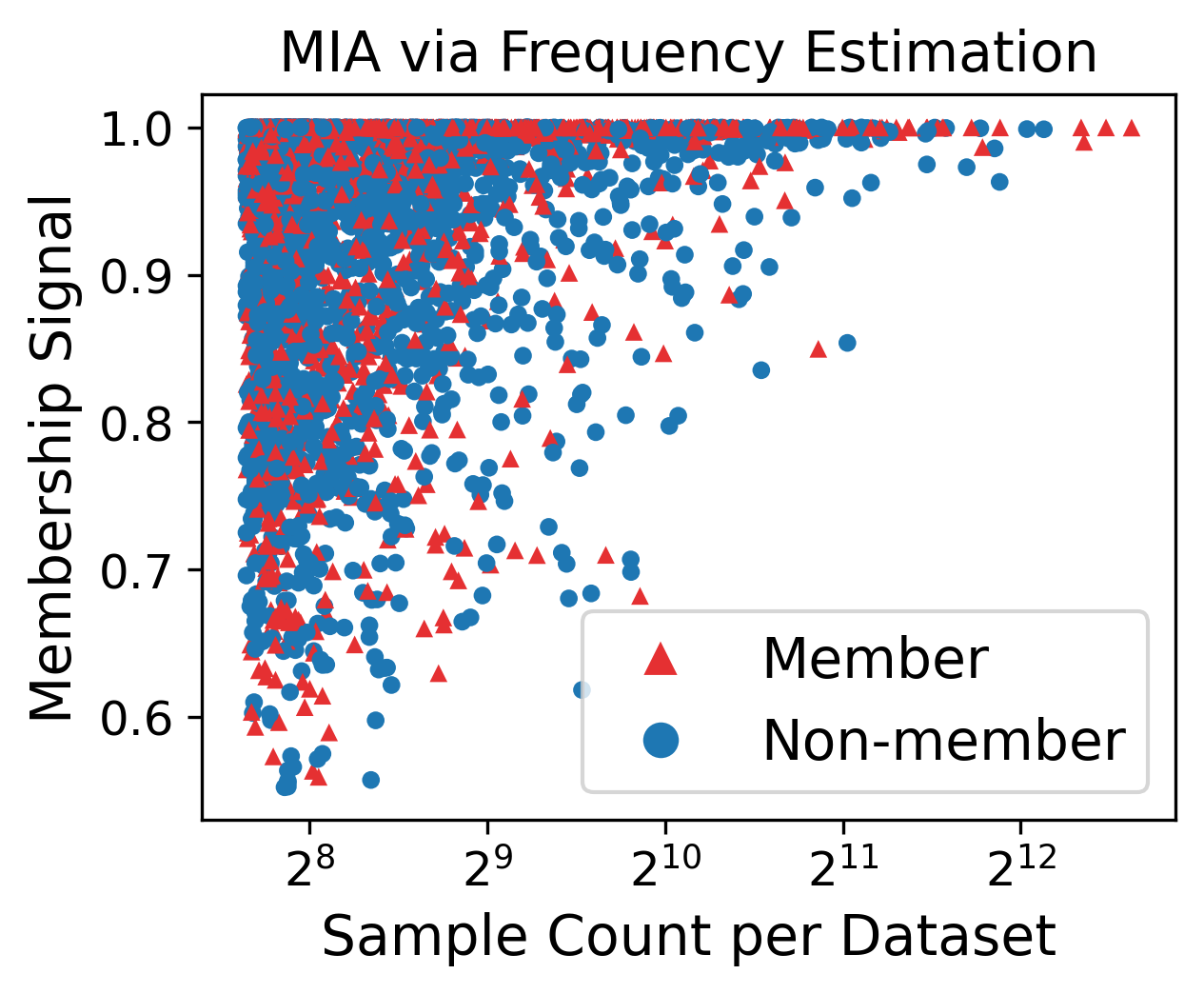}}
    \subfigure[Vocabulary Size: $200\text{,\,}000$]{\includegraphics[width=0.196\textwidth]{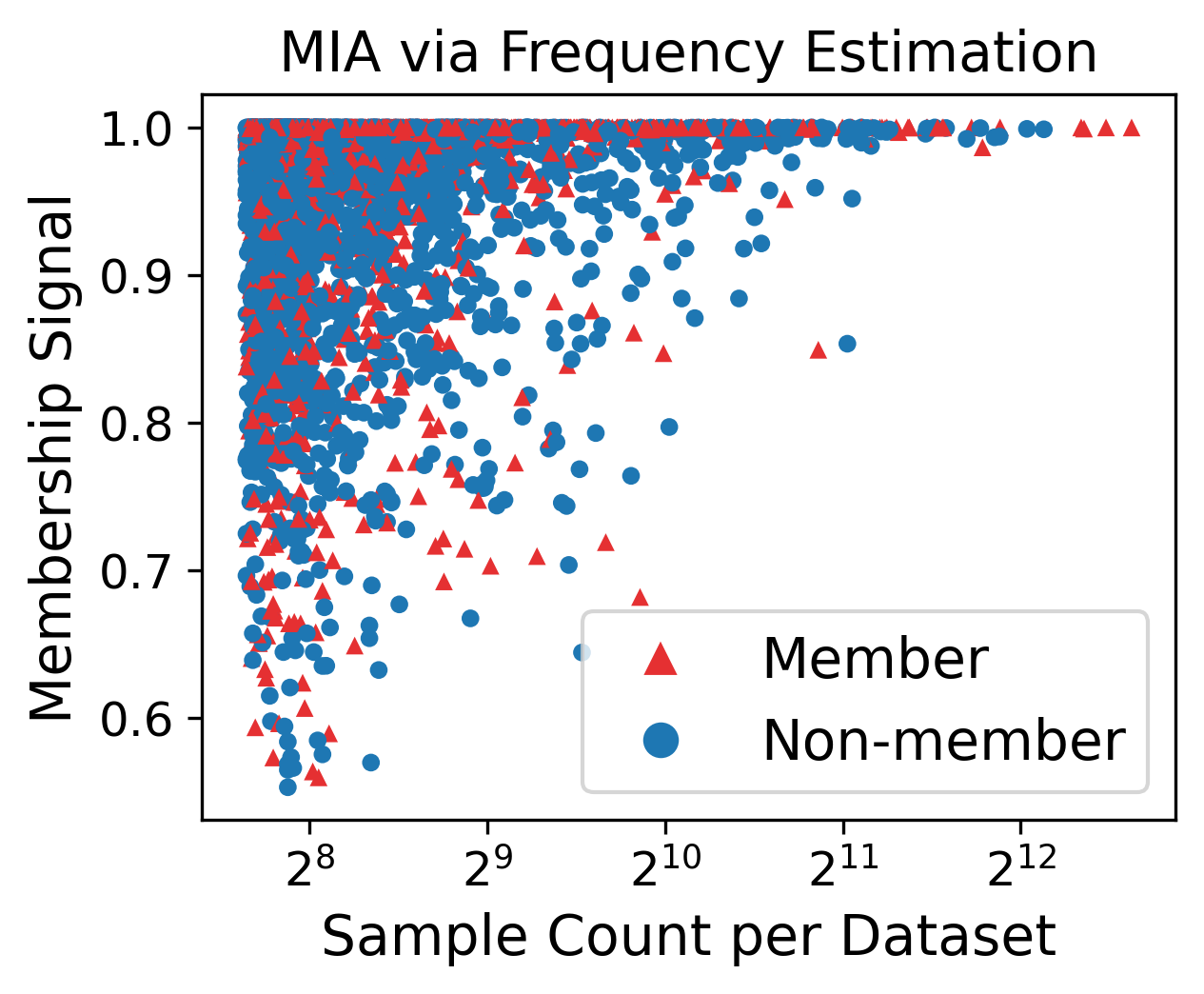}}
    \vspace{-0.1cm}
\caption{Dataset distribution based on MIA via \fsignal.}
    \label{fig: full dataset score F}
\end{figure*}

\cleardoublepage

%% file: table/Dataset_Scores_32.tex
\begin{table*}[t]
    \centering
    \caption{Impact of the target dataset size on defense mechanism ($n_\text{min}=32$).}
    \vspace{-0.1cm}
    \resizebox{\textwidth}{!}{
    \begin{tabular}{@{}l c
                    ccc ccc ccc ccc ccc@{}}
        \toprule
        \multirow{2}{*}{\textbf{Attack Approach}}
        & \multirow{2}{*}{{\#Dataset Size}}
        & \multicolumn{3}{c}{\textbf{$|\mathcal{V}_\text{target}|<80{,}000$}} 
        & \multicolumn{3}{c}{\textbf{$|\mathcal{V}_\text{target}|<110{,}000$}} 
        & \multicolumn{3}{c}{\textbf{$|\mathcal{V}_\text{target}|<140{,}000$}} 
        & \multicolumn{3}{c}{\textbf{$|\mathcal{V}_\text{target}|<170{,}000$}}  
        & \multicolumn{3}{c}{\textbf{$|\mathcal{V}_\text{target}|<200{,}000$}}  \\
        \cmidrule(l{7pt}r{13pt}){3-5} \cmidrule(l{6pt}r{13pt}){6-8} 
        \cmidrule(l{6pt}r{13pt}){9-11} \cmidrule(l{6pt}r{13pt}){12-14} \cmidrule(l{6pt}r{0pt}){15-17}
        & 
        & AUC & BA & TPR
        & AUC & BA & TPR 
        & AUC & BA & TPR 
        & AUC & BA & TPR 
        & AUC & BA & TPR \\
        \midrule
    \multirow{3}{*}{\vsignal}
    & {$|D|\in[0,400)$} 
    & 0.646 & 0.632 & 21.25\% %
    & 0.677 & 0.642 &   21.25\%
    & 0.694 & 0.655  & 23.30\%
    & 0.706 & 0.663  &24.44\%  %
    &  0.717 & 0.669 & 25.14\%  \\%
    
    & {$|D|\in[400,800)$} 
    & 0.703 & 0.667 &  21.25\%%
    & 0.761 & 0.711 & 22.85\% 
    &  0.772 & 0.717  & 33.73\%
    &  0.793 & 0.730  &  34.94\% 
    & 0.810 & 0.742 & 40.06\%  \\

    & {$|D|\in[800,1200)$} 
    & 0.712 & 0.670 &   25.30\% %
    & 0.768 & 0.729 & 30.12\% 
    & 0.798 & 0.761  & 38.75\% 
    & 0.823 &  0.775  &  43.75\%
    & 0.840 & 0.800 &  46.25\% \\

    \midrule
    
    \multirow{3}{*}{\fsignal}
    & {$|D|\in[0,400)$} 
   & 0.586 & 0.591 & 16.36\% 
    & 0.617 & 0.616 & 20.11\%
    & 0.649 & 0.635 & 20.62\%  
    & 0.682 & 0.657 & 21.01\%
    & 0.720 & 0.683 & 26.99\% \\

    & {$|D|\in[400,800)$}
    & 0.611 & 0.615 & 17.78\% 
    & 0.675 & 0.650 & 25.00\%
    & 0.720 & 0.683 & 29.52\%  
    & 0.739 & 0.702 &  30.42\%
    & 0.763  & 0.721 & 33.43\% \\
    
    & {$|D|\in[800,1200)$} 
    & 0.723 & 0.714 & 25.00\% 
    & 0.742 & 0.716   & 35.00\% 
    & 0.758  & 0.729 &   35.00\% 
    & 0.785 & 0.759 &   47.50\% 
    &  0.820 &  0.805 & 52.50\%  \\
    \bottomrule
    \end{tabular}
    }
    \label{tab: dataset score 32}
\end{table*}

%% file: table/Dataset_Scores_48.tex
\begin{table*}[t]
    \centering
    \caption{Impact of the target dataset size on defense mechanism ($n_\text{min}=48$).}
    \vspace{-0.1cm}
    \resizebox{\textwidth}{!}{
    \begin{tabular}{@{}l c
                    ccc ccc ccc ccc ccc@{}}
        \toprule
        \multirow{2}{*}{\textbf{Attack Approach}}
        & \multirow{2}{*}{{\#Dataset Size}}
        & \multicolumn{3}{c}{\textbf{$|\mathcal{V}_\text{target}|<80{,}000$}} 
        & \multicolumn{3}{c}{\textbf{$|\mathcal{V}_\text{target}|<110{,}000$}} 
        & \multicolumn{3}{c}{\textbf{$|\mathcal{V}_\text{target}|<140{,}000$}} 
        & \multicolumn{3}{c}{\textbf{$|\mathcal{V}_\text{target}|<170{,}000$}}  
        & \multicolumn{3}{c}{\textbf{$|\mathcal{V}_\text{target}|<200{,}000$}}  \\
        \cmidrule(l{7pt}r{13pt}){3-5} \cmidrule(l{6pt}r{13pt}){6-8} 
        \cmidrule(l{6pt}r{13pt}){9-11} \cmidrule(l{6pt}r{13pt}){12-14} \cmidrule(l{6pt}r{0pt}){15-17}
        & 
        & AUC & BA & TPR
        & AUC & BA & TPR 
        & AUC & BA & TPR 
        & AUC & BA & TPR 
        & AUC & BA & TPR \\
        \midrule
    \multirow{3}{*}{\vsignal}
    & {$|D|\in[0,400)$} 
    & 0.646 & 0.632 &   21.25\% %
    & 0.677 & 0.642 &     21.25\%
    &  0.687 & 0.645  &  22.85\%
    &0.694  & 0.655  &23.30\%  %
    &   0.710 & 0.667  &  27.12\%  \\%
    
    & {$|D|\in[400,800)$} 
    & 0.704 & 0.667 &   21.58\%%
    & 0.761 & 0.711 &  22.22\% 
    &  0.772 & 0.717  & 30.00\%
    & 0.775  &  0.721   &   34.64\% 
    &  0.790 & 0.730 & 34.94\%  \\

    & {$|D|\in[800,1200)$} 
    & 0.712 & 0.670 &   25.30\% %
    & 0.768 & 0.729 & 30.12\% 
    & 0.774 & 0.738  & 30.72\% 
    &  0.798 &  0.761   &   38.75\%
    &  0.823 &  0.767 &   45.00\% \\

    \midrule
    
    \multirow{3}{*}{\fsignal}
    & {$|D|\in[0,400)$} 
   & 0.586 & 0.591 & 16.36\% 
    & 0.617 & 0.616 & 18.20\%
    & 0.649 & 0.634 & 19.92\%  
    & 0.680 & 0.652 & 21.01\%
    & 0.680 &  0.654 & 22.79\% \\

    & {$|D|\in[400,800)$}
    & 0.588 & 0.597 & 16.93\% 
    & 0.619 & 0.623 &  20.11\%
    & 0.650 & 0.635 &   20.62\%  
    & 0.682 & 0.657 &  21.64\%
    &  0.720  &0.683  &   26.99\% \\
    
    & {$|D|\in[800,1200)$} 
    & 0.616 & 0.611 & 20.48\% 
    & 0.681 & 0.659   & 28.01\% 
    & 0.719  & 0.692 &  28.92\% 
    & 0.731 & 0.693 & 31.93\% 
    &   0.736 &  0.699 & 32.53\%  \\
    \bottomrule
    \end{tabular}
    }
    \label{tab: dataset score 48}
\end{table*}

%% file: table/Dataset_Scores_64_a.tex
\begin{table*}[t]
    \centering
    \caption{Impact of the target dataset size on defense mechanism ($n_\text{min}=64$).}
    \vspace{-0.1cm}
    \resizebox{\textwidth}{!}{
    \begin{tabular}{@{}l c
                    ccc ccc ccc ccc ccc@{}}
        \toprule
        \multirow{2}{*}{\textbf{Attack Approach }}
        & \multirow{2}{*}{{\#Dataset Size}}
        & \multicolumn{3}{c}{\textbf{$|\mathcal{V}_\text{target}|\leq 80{,}000$}} 
        & \multicolumn{3}{c}{\textbf{$|\mathcal{V}_\text{target}|\leq 110{,}000$}} 
        & \multicolumn{3}{c}{\textbf{$|\mathcal{V}_\text{target}|\leq 140{,}000$}} 
        & \multicolumn{3}{c}{\textbf{$|\mathcal{V}_\text{target}|\leq 170{,}000$}}  
        & \multicolumn{3}{c}{\textbf{$|\mathcal{V}_\text{target}|\leq 200{,}000$}}  \\
        \cmidrule(l{7pt}r{13pt}){3-5} \cmidrule(l{6pt}r{13pt}){6-8} 
        \cmidrule(l{6pt}r{13pt}){9-11} \cmidrule(l{6pt}r{13pt}){12-14} \cmidrule(l{6pt}r{0pt}){15-17}
        & 
        & AUC & BA & TPR
        & AUC & BA & TPR 
        & AUC & BA & TPR 
        & AUC & BA & TPR 
        & AUC & BA & TPR \\
        \midrule
    \multirow{3}{*}{\vsignal}
    & {$|D|\in[0,400)$} 
    & 0.646 & 0.632 &19.16\% %
    & 0.662 & 0.636 &  21.25\%
    & 0.677 & 0.642  & 21.25\%
    & 0.683 & 0.647  &23.30\%  %
    &  0.694 & 0.655 & 24.06\%  \\%
    
    & {$|D|\in[400,800)$} 
    & 0.703 & 0.666 &  21.25\%%
    & 0.731 & 0.675 & 21.58\% 
    &  0.751 & 0.699  & 22.85\%
    &  0.761 & 0.711  & 32.50\% 
    & 0.772 & 0.717 & 37.95\%  \\

    & {$|D|\in[800,1200)$} 
    & 0.712 & 0.670 &  25.30\% %
    & 0.743 & 0.702 & 27.11\% 
    & 0.768 & 0.729  & 30.12\% 
    & 0.795 &  0.746  &  33.73\%
    & 0.797 & 0.761 & 38.75\% \\

    \midrule
    
    \multirow{3}{*}{\fsignal}
    & {$|D|\in[0,400)$} 
   & 0.591 & 0.598 & 17.50\% 
    & 0.620 & 0.619 & 18.20\%
    & 0.651 & 0.636 & 20.11\%  
    & 0.654 & 0.637 & 21.07\%
    & 0.656 & 0.639 & 21.39\% \\

    & {$|D|\in[400,800)$}
    & 0.619 & 0.615 & 22.29\% 
    & 0.672 & 0.653 & 28.31\%
    & 0.717 & 0.682 & 30.42\%  
    & 0.718 & 0.683 & 31.93\%
    & 0.723  & 0.686 & 32.83\% \\
    
    & {$|D|\in[800,1200)$} 
    & 0.734 & 0.717 & 26.25\% 
    & 0.739 & 0.731   & 33.75\% 
    & 0.744  & 0.736 &   33.75\% 
    & 0.745 & 0.742 &   35.00\% 
    &  0.748 &  0.746 & 38.75\%  \\
    \bottomrule
    \end{tabular}
    }
    \label{tab: dataset score 64 a}
\end{table*}